\documentclass[10pt, draftcls, onecolumn]{IEEEtran}
\usepackage{amsmath, amssymb,amsthm}
\usepackage{verbatim,makeidx}
\usepackage{algorithm}
\usepackage{algorithmic}
\usepackage{libertine, multirow}
\usepackage{colortbl}
\usepackage{dsfont}
\usepackage{caption}
\usepackage{subcaption}
\usepackage{graphicx}
\usepackage{verbatim}
\usepackage{url}
\usepackage{pbox}
\usepackage{bbm}
\usepackage{authblk}

\newcounter{ALC@tempcntr}
\theoremstyle{plain}
\newtheorem{theorem}{Theorem}
\newtheorem{proposition}{Proposition}
\newtheorem{lemma}{Lemma}
\newtheorem{definition}{Definition}

\theoremstyle{definition}
\newtheorem{example}{Example}

\theoremstyle{remark}
\newtheorem{remark}{Remark}

\newcommand{\beq}{\begin{eqnarray}}
\newcommand{\eeq}{\end{eqnarray}}

\newcommand{\field}[1]{\mathbb{#1}}

\newcommand{\F}{\field{F}}

\newfont{\bbb}{msbm10 scaled 500}

\newfont{\bb}{msbm10 scaled 1100}

\newcommand{\ZZ}{\mbox{\bb Z}}
\newcommand{\FF}{\mbox{\bb F}}

\newcommand{\cv}{{\bf c}}
\newcommand{\dv}{{\bf d}}
\newcommand{\ev}{{\bf e}}
\newcommand{\fv}{{\bf f}}

\newcommand{\mv}{{\bf m}}

\newcommand{\rv}{{\bf r}}
\newcommand{\sv}{{\bf s}}

\newcommand{\xv}{{\bf x}}

\newcommand{\Am}{{\bf A}}

\newcommand{\Cc}{{\cal C}}
\newcommand{\Dc}{{\cal D}}

\newcommand{\Hc}{{\cal H}}
\newcommand{\Ic}{{\cal I}}

\newcommand{\Kc}{{\cal K}}

\newcommand{\Mc}{{\cal M}}

\newcommand{\Pc}{{\cal P}}

\newcommand{\Rc}{{\cal R}}
\newcommand{\Sc}{{\cal S}}

\newcommand{\Uc}{{\cal U}}

\newcommand{\Xc}{{\cal X}}

\newcommand{\Zc}{{\cal Z}}

\newcommand{\remove}[1]{}

\newcommand\ff{{\mathbb F}}

\renewcommand{\arraystretch}{1.3}

\definecolor{OXO-emph}{RGB}{153,0,0}

\newcommand\floorb[1]{\left\lfloor #1 \right\rfloor}

\DeclareMathAlphabet{\mathpzc}{OT1}{pzc}{m}{it}

\theoremstyle{definition}

\theoremstyle{remark}

\newcommand{\latexe}{{\LaTeX\kern.125em2%
                      \lower.5ex\hbox{$\varepsilon$}}}

\chardef\bslash=`\\	

\makeatletter		
\def\square{\RIfM@\bgroup\else$\bgroup\aftergroup$\fi
\vcenter{\hrule\hbox{\vrule\@height.6em\kern.6em\vrule}
\hrule}\egroup}\makeatother\makeindex

\begin{document}
\sloppy

\title{Centralized Repair of Multiple Node Failures with Applications to Communication Efficient Secret Sharing}


\author{Ankit~Singh~Rawat, O.~Ozan~Koyluoglu, and~Sriram~Vishwanath
\thanks{This paper was presented in parts at IEEE Information Theory and Applications Workshop, San Diego, CA, February 2016.}
\thanks{A.~S.~Rawat is with the Computer Science Department, Carnegie Mellon University, Pittsburgh, PA 15213 USA (e-mail: {asrawat@andrew.cmu.edu}).}
\thanks{O.~O.~Koyluoglu is with the Department of Electrical and Computer Engineering, The University of Arizona, Tucson, AZ 85721 USA (e-mail: ozan@email.arizona.edu).}
\thanks{S.~Vishwanath is with the Laboratory of Informatics, Networks and Communications, Department of Electrical and Computer Engineering, The University of Texas at Austin, Austin, TX 78751 USA (e-mail: sriram@austin.utexas.edu).}
}

\allowdisplaybreaks
\maketitle



\begin{abstract}
This paper considers a distributed storage system, where multiple storage nodes can be reconstructed simultaneously at a centralized location. This centralized multi-node repair (CMR) model is a generalization of regenerating codes that allow for bandwidth-efficient repair of a single failed node. This work focuses on the trade-off between the amount of data stored and repair bandwidth in this CMR model. In particular, repair bandwidth bounds are derived for the minimum storage multi-node repair (MSMR) and the minimum bandwidth multi-node repair (MBMR) operating points. The tightness of these bounds are analyzed via code constructions. The MSMR point is characterized through codes achieving this point under functional repair for general set of CMR parameters, as well as with codes enabling exact repair for certain CMR parameters. The MBMR point, on the other hand, is characterized with exact repair codes for all CMR parameters for systems that satisfy a certain entropy accumulation property. Finally, the model proposed here is utilized for the secret sharing problem, where the codes for the multi-node repair problem is used to construct communication efficient secret sharing schemes with the property of bandwidth efficient share repair.
\end{abstract}




\section{Introduction}

The ability to preserve the stored information and maintain the seamless operation in the event of permanent failures and (or) transient unavailability of the storage nodes is one of the most important issues that need to be addressed while designing distributed storage systems. This gives rise to the so called `code repair' or `node repair' problem which requires a storage system to enable mechanism to regenerate (repair) the content stored on some (failed/unavailable) storage nodes with the help of the content stored on the remaining (live/available) nodes in the system. A simple replication scheme where one stores multiple copies of each data block on different nodes clearly enables the node repair as one can regenerate the data blocks stored on a node by obtaining  one of their copies from the other nodes in the system. However, replication suffers from the decreasing rate as one increases the replication factor in order to enhance the resilience of the system. This motives the use of erasure codes as they efficiently trade-off the storage space for the ability to tolerate failure/unavailability of storage nodes. However, the better utilization of the storage space should also be accompanied by a resource-efficient node repair process and efficiency of the node repair becomes a yardstick for implementing one erasure code over another. 

Towards this, Dimakis et al. propose repair bandwidth, the amount of data downloaded from the contacted nodes during the repair of a single node, as a measure of the efficiency of the repair process in \cite{dimakis}. Considering $n$ storage nodes where any set of $k$ nodes are sufficient to reconstruct the entire information, Dimakis et al. further characterize an information-theoretic trade-off among the storage space vs. the repair bandwidth for such codes. The codes which attain any point on this trade-off are referred to as regenerating codes. Over the past few years, the problem of designing regenerating codes has fueled numerous research efforts which have resulted into the constructions presented in \cite{dimakis, RSK11, zigzag13, PapDimCad_hadamard, SAK15} and the references therein.

In this paper, we explore the problem of enabling bandwidth efficient repair of multiple nodes in a centralized manner. In particular, we consider a setting where one requires the content of any $k$ out of $n$ nodes in the system to be sufficient to reconstruct the entire information (as a parameter for the worst case fault-tolerance of the system). As for the centralized repair process, we consider a framework where the repair of $t \geq 1$ node failures is performed by contacting any $d$ out of the $n - t$ remaining storage nodes. We also assume that $\beta$ amount of data from each of the $d$ contacted nodes are downloaded. We aim to characterize the storage vs. repair-bandwidth trade-off under this centralized multi-node repair (CMR) framework.

We believe that this framework is more suitable for the setting of large scale storage systems where there is a need to perform repairs at a central location. Our CMR model is perhaps useful for the following scenarios: \noindent \textbf{a) Architectural and implementation related issues:} Architectural constraints could make it more efficient to regenerate the content in a centralized manner. For instance, in a rack-based node placement architecture, a top-of-the-rack (TOR) switch failure would imply failure of nodes in the corresponding rack to be unaccessible, and regenerating entire content of the failed rack on a per-node basis, i.e., independently one by one, would be less efficient as compared to regenerating the content at a central location, e.g., at a leader node in that rack.
\noindent \textbf{b) Threshold-based data maintenance:} These schemes regenerate servers after a threshold number of them fail. After regenerating the content stored on the failed nodes, the administrator can recruit $t$ newcomers as replacements of the failed nodes and re-distribute the data to the newcomers in order to restore the state of the system prior to the failures.
\noindent \textbf{c) Availability:} In the event of transient unavailability of the $t$ storage nodes, the centralized repair process allows the user to get access the content stored on the unavailable nodes in a bandwidth efficient manner.

\subsection{Related work} 
We note that the repair of multiple nodes in a bandwidth-efficient manner has previously been  considered under the cooperative repair model introduced in \cite{ShumHu, Kermarrec:Repairing11}. There are two major differences between the cooperative and centralized repair frameworks: a) Under cooperative repair framework~\cite{ShumHu, Kermarrec:Repairing11}, all $t$ newcomer nodes are not constrained to contact the same set of $d$ out of $n-t$ surviving nodes. The framework allows each newcomer to contact any $d$ surviving nodes independent of the nodes contacted by other $t-1$ newcomers. b) Under cooperative repair framework, after downloading data from the surviving nodes, the newcomers exchange certain amount of data among themselves. On the other hand, since a centralized entity (e.g., the administrator or a master server node) has access to all the downloaded information, such information exchange is not required in the centralized repair model. Our hope is that removing the additional restriction imposed by the cooperative repair framework will enable designing codes for a broader range of system parameters. 

The problem of centralized bandwidth-efficient repair of multiple node failures in a DSS employing has previously been considered by Cadambe et al.~\cite{CadambeT13}. However, they restrict themselves to only MDS codes and they show existence of such codes only in the asymptotic regime where node size (amount of data stored on a node) tends to infinity.

In addition, locality, the number of nodes contacted during repair of a single node, is another measure of node repair efficiency which have been extensively studied in the literature~\cite{Gopalan12}. Various minimum distance bounds and constructions achieving trade-offs are presented in \cite{Gopalan12, RKSV12, KPLK12, TamoBarg} and the references therein. In particular, recent works \cite{RMV15, KumarTwo,Song2015} have studied locality problem with multiple node repairs, which is a model relevant to the framework studied in this paper.

Finally, in a recent work \cite{HLKB15}, Huang et al. proposed a model for {\em communication efficient secret sharing}, where the system stores a secret over $n$ nodes (shares) with the property that accessing to any $z$ shares does not reveal any information about the secret, and accessing to any $d$ shares does reveal the secret. The framework \cite{HLKB15} is similar to that of \cite{dimakis} in the sense that one contacts to more than enough number of nodes (and download a partial  data from each) in order to reduce the total amount of bits downloaded (to reveal secret in the former, and to repair a node in the latter). Given this setup, \cite{HLKB15} provides a bound on required amount of communication to reconstruct the secret, constructs explicit coding schemes for certain parameter regimes achieving the stated bound, and shows an existence result for general set of parameters. More recently, \cite{BR15} focuses on the same model and proposes codes that can achieve the bound provided in \cite{HLKB15} for general set of parameters. A separate body of work \cite{PRR11, SRK_globecom11, RKSV12, HPX15a, KRV12, HPX15b} considers secure regenerating codes, where eavesdropper accessing to a subset of nodes in the system does not get any information about the stored data in the system. These works essentially focus on characterizing the maximum amount of secret bits that can be stored within a system that employs a given regenerating code (e.g., MSR/MBR). In this sense, these works consider a storage of data that is composed of both public (without security constraints) and private (with security constraints) information, and a data collector connects to a predefined number of nodes to recover both types of information. Whereas, in \cite{HLKB15}, only the reconstruction of the private information is the concern. In addition to this key difference, the eavesdropper models in secure regenerating code papers also include eavesdroppers that can observe the data transferred during node repairs\footnote{This eavesdropping model is important for non-MBR codes, as for MBR codes, the amount of downloaded content for a node repair is same as the data stored in the node.}, whereas the framework in \cite{HLKB15} does not consider repair problem. We note that regenerating coding schemes which are secure against such eavesdroppers are presented in \cite{PRR11, SRK_globecom11, RKSV12, HPX15a} and references therein. And, the problem of designing secure cooperative regenerating codes is explored in \cite{KRV12, HPX15b}.

\subsection{Contributions} 
The results of this work can be summarized as follows.
\begin{itemize}
\item We develop general repair bandwidth bounds for the CMR model at minimum per-node storage multi-node repair and minimum bandwidth multi-node repair regimes, referred to as MSMR and MBMR operating points respectively. 
\item We investigate tightness of the derived bounds with appropriate code constructions, and characterize the fundamental limits of the CMR model. In particular, for the MSMR scenario, the fundamental limit is characterized utilizing functional repair for all parameters. For special cases, explicit constructions that achieve the stated bound are also provided. These constructions are based on cooperative regenerating codes with minimum per-node storage (MSCR) codes as well as Zigzag codes. For the former set of codes, we show a result that any MSCR code can be utilized as MSMR code achieving the stated bounds. For the latter case, we show that multiple nodes can be repaired in Zigzag codes, and this proposed repair process is bandwidth-wise optimal, achieving the derived bound in this paper.
\item For the MBMR scenario, we define minimum repair bandwidth as the property of having amount of downloaded data matching to the entropy of $t$ nodes. For this setup, the fundamental limit is characterized for systems having a certain entropy accumulation property.  In addition, we obtain a general mapping from minimum bandwidth cooperative regenerating (MBCR) codes  to MBMR codes, and, utilizing MBCR with a certain entropy accumulation property, we show achievability of the stated bounds, characterizing the MBMR operating point in this special entropy accumulation case.
\item Finally, we focus on the secret sharing problem, and show that the codes for the multi-node repair problem can be transformed into communication efficient secret sharing schemes that posses not only the reliability (for multi-node repairs) but also the security properties. We propose a secret sharing mechanism with repairable shares that have the highest possible repair bandwidth-efficiency in the multi-node failure setup. Adversarial attack setup is considered to provide secrecy. 
\end{itemize}

\section{Centralized multi-node repair model}

We introduce a new model for simultaneous repair of multiple node failures in a distributed storage system (DSS), namely {\em centralized multi-node repair (CMR) model}. Consider an $(n, k)$-DSS, {\em i.e.}, the system comprises $n$ storage nodes and the content stored on any $k$ nodes is sufficient to reconstruct the information stored on the system. For an $(n, k)$-DSS, under $(d, t)$-CMR model, any set of $t$ failed nodes in the system can be repaired by downloading data from any set of $d$ out of $n - t$ surviving nodes. Let $\alpha$ denote the size of each node (over a finite field $\FF$) and $\beta$ denote the amount of data downloaded from each of the contacted $d$ nodes under the $(d, t)$-CMR model. In order to denote all the relevant system parameters, we also expand the notation for the CMR model as $(n, k, d, t, \alpha, \gamma)$-CMR model or $(d, t, \alpha, \gamma)$-CMR model. After downloading $\gamma = d\beta$ symbols from the contacted nodes, the content stored on all $t$ failed nodes is recovered simultaneously in a centralized manner\footnote{The CMR model also allow for the distributed/parallel repair of all the $t$ failed nodes by $t$ newcomers independently. However, it is assumed that each of the $t$ newcomers have an access to all the $\gamma$ downloaded symbols.}.

\section{A file size bound for the CMR model}
\label{sec:gen}

In this section, we initiate the study of the trade-off between the per-node storage $\alpha$ and repair bandwidth $\gamma$ for the CMR model. We first provide a file size bound for the CMR model.

Let the system store a uniformly distributed file $\fv$ of size $|\fv|=\Mc$ (over a finite field $\F$). Consider the case when the nodes indexed by a set $\Kc \subseteq [n]$ such that $|\Kc|=k$ are used to reconstruct the file $\fv$. Further, assume that this set of nodes are partitioned into $g$ number of distinct subsets $\Sc_i$ with $|\Sc_i|=n_i\leq t$ such that $\sum\limits_{i=1}^{g} n_i=k$. We have the following bound.

\begin{lemma}\label{thm:FileSizeBound}
The system parameters necessarily satisfy
\begin{equation}
\Mc 
\leq \sum\limits_{i=1}^g \min\Big\{n_i\alpha, \big(d-\sum\limits_{j=1}^{i-1}n_j\big)\beta\Big\}. \label{eq:bound}
\end{equation}
\end{lemma}

\begin{proof}
Denoting the symbols stored on the nodes indexed by the set $\Sc$ by $\xv_{\Sc}$, we  have 
\begin{align}
\Mc &= H(\fv) \stackrel{(a)}{=} H(\fv)-H(\fv|\xv_{\Kc}) = I(\xv_{\Kc};\fv) \leq H(\xv_{\Kc}) \\
&\stackrel{(b)}{=}\sum\limits_{i=1}^g H(\xv_{\Sc_i}|\xv_{\Sc_1:\Sc_{i-1}}) \\
&\stackrel{(c)}{\leq} \sum\limits_{i=1}^g \min\left\{H(\xv_{\Sc_i}), \big(d-\sum\limits_{j=1}^{i-1}n_j\big)\beta\right\} \label{eq:bound_h}\\
&\stackrel{(d)}{\leq} \sum\limits_{i=1}^g \min\left\{n_i\alpha, \big(d-\sum\limits_{j=1}^{i-1}n_j\big)\beta\right\},
\end{align}
where (a) is due to recoverability constraint $H(\fv|\xv_{\Kc})=0$ as $|\Kc|=k$, 
(b) is due to $\Kc=\cup_{i=1}^g \Sc_i$,
(c) \& (d) are due to the following bounds for each term in the sum: $H(\xv_{\Sc_i}|\xv_{\Sc_1:\Sc_{i-1}})\leq H(\xv_{\Sc_i})\leq n_i\alpha$, and
\begin{align*}
H(\xv_{\Sc_i}|\xv_{\Sc_1:\Sc_{i-1}})
&\stackrel{(e)}{=} H(\xv_{\Sc_i}|\xv_{\Sc_1:\Sc_{i-1}}) \nonumber\\
&\quad - H(\xv_{\Sc_i}|\xv_{\Sc_1:\Sc_{i-1}}, \dv_{\Hc_i-\Sc_{1}:\Sc_{i-1}}) \\
&= I(\dv_{\Hc_i-\Sc_1:\Sc_{i-1}}; \xv_{\Sc_i} | \xv_{\Sc_1:\Sc_{i-1}}) \\
&\leq H(\dv_{\Hc_i-\Sc_1:\Sc_{i-1}}) \leq \left(d-\sum\limits_{j=1}^{i-1} n_j \right) \beta
\end{align*}
where set of helper nodes to regenerate symbols in $\Sc_i$ is denoted as $\Hc_i$, this set of $d$ nodes is constructed by using the sets $\Sc_{1} \cdots \Sc_{i-1}$ and additional nodes not belonging to these sets (this is possible as $\sum\limits_{i=1}^{g} n_i=k\leq d$), downloaded symbols from these additional nodes are denoted as $\dv_{\Hc_i-\Sc_1:\Sc_{i-1}}$ with $|\Hc_i-\Sc_1:\Sc_{i-1}|=d-\sum\limits_j^{i-1} n_j$, and (e) follows as $H(\xv_{\Sc_i}|\xv_{\Sc_1:\Sc_{i-1}}, \dv_{\Hc_i-\Sc_{1:i-1}})=0$ as $H(\xv_{\Sc}|\dv_{\Hc})$=0 for any $\Sc$ such that $|\Sc|\leq t$ and any $\Hc$ such that $|\Hc|=d$.
\end{proof}

Given the bound in Proposition~\ref{thm:FileSizeBound}, we differentiate between two operating regimes of the system: {\em Minimum storage multi-node regeneration (MSMR)} and {\em minimum bandwidth multi-node regeneration (MBMR)}. The MSMR point corresponds to having an MDS code which requires that $\alpha=\Mc/k$. Codes that attain minimum possible repair bandwidth under this constraint, i.e., $\alpha=\Mc/k$,   are referred to as MSMR codes. On the other hand, the MBMR point restricts that $H(\xv_{\Sc})  = \gamma = d\beta$ for every $\Sc \subseteq [n]$ such that $|\Sc| = t$, i.e., the amount of data downloaded during the centralized repair of $t$ node failures is equal to the amount of information stored on the lost $t$ nodes. MBMR codes achieve the minimum possible repair bandwidth under this restriction, i.e., $H(\xv_{\Sc})  = \gamma = d\beta$. In the following, we focus on the problem of characterizing these two operating points of the CMR model.

\section{MSMR Codes}
\label{sec:msmr_codes}

We first utilize Lemma~\ref{thm:FileSizeBound} to obtain a bound on the repair bandwidth at the MSMR point, and then focus on achievability.

\subsection{Repair bandwidth bound}

\begin{proposition}
\label{prop:MSMR}
Consider an $(n, k)$-DSS that stores a file of size $\Mc$ 
and enables repair of $t$ failed nodes under a $(d, t, \alpha_{MSMR} = \frac{\Mc}{k}, \gamma)$-CMR model. Then, we have 
\begin{align}
\label{eq:MSMR}
\gamma_{MSMR} \geq \frac{\Mc d t}{k(d - k + t)}.
\end{align}
\end{proposition}

\begin{proof}
Let $a=\lfloor k/t \rfloor$ and $b=k-at$. We set $n_1 = b$ and $n_i=t$ for $i=2,\cdots,g = a+1$. From the bound \eqref{eq:bound}, we obtain
\begin{align}
\Mc&\leq 
\min\left\{b\alpha, d\beta\right\} + \sum\limits_{i=1}^{a} \min\left\{
t\alpha, [d-(i-1)t - b]\beta
\right\}.
\end{align}
Note that we have $\alpha=\frac{\Mc}{k}$ which implies that $d\beta \geq b\alpha$ and $$[d-(i-1)t - b]\beta \geq t\alpha, \forall i=1,\cdots,a,$$ From this, we obtain $\beta \geq \frac{b\alpha}{d}$ and $[d-(a-1)t - b]\beta\geq t\alpha$, i.e., $\beta\geq \frac{t\alpha}{[d-at - b+t]} = \frac{t\alpha}{[d-k+t]}$. This implies that 
\begin{align}
\gamma_{MSMR}&= d\beta\geq d\alpha \max\left\{\frac{t}{d- k + t},\frac{b}{d}\right\} 
\overset{(i)}{=} \frac{\Mc d t}{k(d - k + t)}, \nonumber 
\end{align}
where $(i)$ follows from the fact that we have $b<t\leq k$ and $\alpha = \frac{\Mc}{k}$.
\end{proof}
\begin{remark}
Note that the same bound is also obtained by Cadambe et al. in \cite{CadambeT13} where they consider repair of multiple failures in an MDS code.
\end{remark}
\begin{remark}\label{rem:MSMR}
A code that allows for repair of $t$ failed nodes with the parameters $\big(d, t, \alpha = \frac{\Mc}{k}, \gamma  = \frac{\Mc d t}{k(d - k + t)}\big)$-CMR is an MSMR code. 
\end{remark}

\begin{proposition}
The bound above \eqref{eq:MSMR} does not improve when helper nodes are allowed to contribute different amounts of data for regeneration of $t$ nodes.
\end{proposition} 

\begin{proof}
The proof follows from the steps given in \cite{HLKB15}.
Assume that the $n$ nodes in the DSS are indexed by the set $[n]$. Let's consider a specific failure pattern, where the $t$ nodes indexed by the set $[t] \subset [n]$ are under failure. Furthermore, we assume that the $d$ nodes indexed by the set $\{t+1, t+2,\ldots, t + d\}$ are contacted to repair the $t$ failures under the centralized repair model. For $j \in \{t+1, t+2,\ldots, t+d\}$, let $\sv_{j}$ denote the symbols downloaded from the node indexed by $j$ in order to repair the $t$ failed nodes. Without loss of generality, we can assume that\footnote{Note that the proof holds even when we define $\beta_t = \frac{\gamma_t}{d} = \frac{\sum_{j = t+1}^{t+d}|\sv_j|}{d}$, i.e., $\beta_t$ represents the average number of symbols downloaded from each of the contacted nodes. In the special setting where we have each contacted node contributes the equal number of symbols during the centralized node repair process, we have $\beta_t = |\sv_{t+1}| = \cdots = |\sv_{t+d}|$.}
\begin{align}
\label{eq:msr_CR1}
|\sv_{t+1}| \geq |\sv_{t+2}| \geq \cdots \geq |\sv_{t+d}|.
\end{align}
Note that an $(n, k)$-coding scheme with $\alpha = \frac{\Mc}{k}$ is an MDS coding scheme. Therefore, the content of the nodes indexed by the set $\{t+1,\ldots, k\}$ does not provide any information about the content of the failed nodes, i.e., the nodes indexed by the set $[t]$. Therefore, in order to be able to repair the $t$ failed nodes, we need to have
\begin{align}
\label{eq:msr_CR2}
\sum_{j = t+k+1}^{t+d}|\sv_i| \geq t\alpha,
\end{align}
i.e., the amount of data downloaded from the remaining $d + t - k$ contacted nodes should be at least the amount of information lost due to node failures. Therefore, we have 

\begin{align}
\label{eq:msr_CR3}
\gamma_t = \sum_{j = t+d}^{t+k}|\sv_j| &\overset{(a)}{\geq} \frac{d}{d - k + t}\sum_{j = t+k+1}^{t+d}|\sv_j| \nonumber \\
&\overset{(b)}{\geq} \frac{dt\alpha}{d - k + t},
\end{align}
where (a) and (b) follow from \eqref{eq:msr_CR1} and \eqref{eq:msr_CR2}, respectively.

\end{proof}

\subsection{Constructions and the characterization of the MSMR point}

\begin{figure*}
\begin{center}
   \small
    \begin{tabular}{| c | c | c | c | c | c |}
    \hline
    $1$&$2$&$3$&$4$&$5$ & $6$ \\ \hline
    $x_{0,0}$ &$x_{0, 1}$&${\color{green} x_{0,2}}$&${\color{green}x_{0,0} + x_{0,1} + x_{0,2}}$&${\color{red}x_{0,0} + x_{6,1} + x_{2,2}}$ & ${\color{red}x_{0,0} + x_{3,1} + x_{1,2}}$ \\ \hline
   $x_{1, 0}$ &$x_{1,1}$&${\color{red}x_{1,2}}$&${\color{red} x_{1,0} + x_{1,1} + x_{1,2}}$&${\color{green}x_{1,0} + x_{7,1} + x_{0,2}}$ & ${\color{red}x_{1,0} + x_{4,1} + x_{2,2}}$ \\ \hline
   $x_{2, 0}$ &$x_{2, 1}$ &${\color{red}x_{2,2}}$&${\color{red} x_{2,0} + x_{2,1} + x_{2,2}}$&${\color{red}x_{2,0} + x_{8,1} + x_{1,2}}$ & ${\color{green} x_{2,0} + x_{5,1} + x_{0,2}}$ \\ \hline
  $x_{3, 0}$ &$x_{3, 1}$ &$x_{3,2}$&$x_{3,0} + x_{3,1} + x_{3,2}$&${\color{blue}x_{3,0} + x_{0,1} + x_{5,2}}$ & $x_{3,0} + x_{6,1} + x_{4,2}$ \\ \hline
    $x_{4, 0}$ &$x_{4, 1}$ &$x_{4,2}$&$x_{4,0} + x_{4,1} + x_{4,2}$&$x_{4,0} + x_{1,1} + x_{3,2}$ & ${\color{blue} x_{4,0} + x_{7,1} + x_{5,2}}$ \\ \hline
    $x_{5, 0}$ &$x_{5, 1}$ &${\color{blue}x_{5,2}}$&${\color{blue}x_{5,0} + x_{5,1} + x_{5,2}}$&$x_{5,0} + x_{2,1} + x_{4,2}$ & $x_{5,0} + x_{8,1} + x_{3,2}$ \\ \hline
    $x_{6, 0}$ &$x_{6, 1}$&${\color{magenta} x_{6,2}}$&${\color{magenta}x_{6,0} + x_{6,1} + x_{6,2}}$&$x_{6,0} + x_{3,1} + x_{8,2}$ & ${\color{blue} x_{6,0} + x_{0,1} + x_{7,2}}$ \\ \hline
    $x_{7, 0}$ &$x_{7, 1}$ &${\color{blue}x_{7,2}}$&${\color{blue}x_{7,0} + x_{7,1} + x_{7,2}}$&${\color{magenta}x_{7,0} + x_{4,1} + x_{6,2}}$ & $x_{7,0} + x_{1,1} + x_{8,2}$ \\ \hline
    $x_{8, 0}$ &$x_{8, 1}$&$x_{8,2}$&$x_{8,0} + x_{8,1} + x_{8,2}$&${\color{blue}x_{8,0} + x_{5,1} + x_{7,2}}$ & ${\color{magenta}x_{8,0} + x_{2,1} + x_{6,2}}$ \\ \hline
   \end{tabular}
\end{center}
\caption{Repair of the first two systematic nodes in a $(6, 3)$-zigzag code. (Coding coefficients of the parity symbols are not specified.) Blue (red) colored symbols contribute in the repair of only node $1$ (respectively, $2$) in the case of single node failure. Green colored symbols contribute in the repair of both node $1$ and node $2$ in the case of single node failure. Magenta colored symbols denote the additional symbols that need to be downloaded to enable the centralized repair of both the nodes.}
\label{tab:example1_repair}
\end{figure*}

\subsubsection{Constructions from existing MSCR codes}
\label{sec:msmr1}

Minimum storage cooperative regenerating (MSCR) codes allow for simultaneous repair of $t$ storage nodes with the following scheme: Each newcomer node contacts to $d$ nodes and downloads $\beta$ symbols from each. (Different nodes can contact to different live nodes.) Then, each newcomer node sends $\beta'$ symbols to each other. Under this setup, the repair bandwidth \emph{per failed node} is $d\beta+(t-1)\beta'$. MSCR codes operate at $\alpha_{MSCR}=\Mc/k$ and $\beta_{MSCR}=\beta'_{MSCR}=\frac{\Mc}{k(d-k+t)}$.

\begin{proposition}
\label{prop:fromMSCR}
A code $\Cc$ that operates as an MSCR code is also an MSMR code for the CMR model.
\end{proposition}
\begin{proof}
Consider that each failed node contact to the same set of $d$ nodes in the MSCR code $\Cc$. Then, each failed node downloads $\beta_{MSCR}$ symbols from these $d$ helper nodes, resulting in a total of at most $\gamma=td\beta_{MSCR}=\frac{\Mc dt}{k(d-k+t)}$ symbols. These symbols can recover each failed node, hence regenerates $t$ failed nodes in the CMR model. Therefore, code $\Cc$ is an MSMR code with $\alpha=\frac{\Mc}{k}$ and $\gamma=\frac{\Mc dt}{k(d-k+t)}$.
\end{proof}
We remark that random linear network coding attains MSCR point \cite{ShumHu}, hence it provides an MSMR code with functional repair. Explicit code constructions for the MSCR setup while ensuring exact-repair, on the other hand, are know for a small set of parameters. The only such constructions that we are aware of are provided in \cite{LeScouarnec:Exact12} for $k=t=2$, in \cite{LiBao14} for $t=2$ (for parameters $(n,k,d)$ at which $(n, k, d+1)$ MSR codes exist), and in \cite{ShumHu} for $d=k$. We believe that moving from the cooperative repair model~\cite{ShumHu, Kermarrec:Repairing11} to the CMR model would allow us to construct MDS codes (MSMR codes) that enable repair-bandwidth efficient repair of $t$ nodes for an expanded set of system parameters.  We exhibit this by designing a scheme to perform centralized repair of multiple nodes in a distributed storage system employing a zigzag code~\cite{zigzag13}.

\subsubsection{Centralized repair of multiple node failures in a zigzag code~\cite{zigzag13}}
\label{sec:msmr2}

The zigzag codes, as introduced in \cite{zigzag13}, are MDS codes that allow for repair of a single node failure among systematic nodes by contacting $d = n - 1$ (all of the) remaining nodes. The zigzag codes are associated with the MSR point~\cite{dimakis} (or MSMR point with $t = 1$ (cf.~\eqref{eq:MSMR})) as each of the contacted $d = n-1$ nodes contributes $\beta = \frac{\alpha}{d - k + 1} = \frac{\alpha}{n-k}$ symbols during the repair of a single failed node. This amounts to the repair bandwidth of $\gamma = d\beta = \frac{n - 1}{n-k}\alpha$. Here, we show that the framework of zigzag codes also enable repair of multiple nodes in the CMR model. 

We state the achievable parameters in the following result. We then illustrate the proposed centralized repair scheme with the help an example of an $(n = 6, k = 3)$-zigzag code where we can simultaneously repair any $2$ systematic nodes\footnote{In a parallel and independent work~\cite{WTB16}, the authors present a mechanism for repairing multiple failures in zigzag codes as well. They show that the zigzag codes can repair any $t \leq n - k$ failures while achieving the lower bound in \eqref{eq:MSMR}.}. 
\begin{theorem}
\label{thm:zigzag}
For an $(n = k  + r, k)$ zigzag code with $r = n - k \geq 2$, it is possible to repair any $1\leq t \leq 3$ systematic nodes in a centralized manner with the optimal repair-bandwidth (cf.~\ref{eq:MSMR}) by contacting $d = n - t$ helper nodes.
\end{theorem}
\begin{proof}
We provide the details of the repair process for $2\leq t \leq 3$ systematic nodes along with the necessary background on zigzag codes in Appendix~\ref{appen:zigzag}.
\end{proof}

\begin{example}[Repairing $t = 2$ systematic nodes in a $(6, 3)$-zigzag code]
Let's consider a zigzag code with the parameters $n = 6, k = 3$ and $\alpha = 9$ from \cite{zigzag13}. This code is illustrated in Table~\ref{tab:example1_repair} where each column (indexed from $1$ to $6$) represents a storage node. Recall that, in the event of a single node failure, this code allows for the repair of any systematic node failure by contacting $\hat{d} = 5$ remaining nodes and downloading $\beta = \frac{\alpha}{n - k} = 3$ symbols from each of these nodes. We now show that we can use this same construction (with required modifications of the non-zero coefficients in coded symbols) to repair $2$ systematic node failures by contacting ${d} = n - 2 = 4$ remaining nodes. We download $t\frac{\alpha}{\widehat{d} - k + 2} = 2\frac{\alpha}{n - k} = 6$ symbols from each of the $d=4$ contacted nodes. 

Assume that node $1$ and $2$ are in failure. We download the colored symbols from node $3$ to node $6$ in Figure~\ref{tab:example1_repair} to repair these two nodes. Using the downloaded symbols, we get the following $18$ combinations in the $18$ unknown information symbols. (We suppress the coefficients of the linear combinations here.)
\begin{align}
\label{eq:18_eq}
&{\color{red} x_{0,0} }+ x_{6, 1},~{\color{red} x_{1,0}} + x_{4, 1},~{\color{red} x_{2,0}} + x_{2, 1},~{\color{red} x_{3,0}} + x_{0, 1}, \nonumber \\
& {\color{red} x_{4,0}} + x_{7, 1}, {\color{red} x_{5,0}} + x_{5, 1},~{\color{red} x_{6,0}} + x_{0, 1},~{\color{red} x_{7,0}} + x_{7, 1},\nonumber \\
&{\color{red} x_{8,0}} + x_{5, 1}, x_{2,0} +{\color{red}  x_{8, 1}},~x_{1,0} + {\color{red} x_{7, 1}},~x_{6,0} + {\color{red} x_{6, 1}},\nonumber \\
&x_{2,0} + {\color{red} x_{5, 1}},~x_{7,0} + {\color{red} x_{4, 1}}, x_{0,0} + {\color{red} x_{3, 1}},~x_{8,0} + {\color{red} x_{2, 1}},\nonumber \\
&x_{1,0} + {\color{red} x_{1, 1}},~x_{0,0} +{\color{red}  x_{0, 1}}.
\end{align}

Now, we need to show that it is possible to choose the coding coefficients in such a manner that these $18$ equations allow us to recover the desired $18$ symbols. Assuming that $A$ denotes the $18 \times 18$ coefficient matrix of the aforementioned $18$ combinations, it is a necessary and sufficient (with large enough field size) condition for the matrix $A$ to be full rank that the natural bipartite graph associated with the matrix $A$ contains a perfect matching\footnote{The left and the right nodes in the bipartite graph correspond to the combinations and the unknowns, respectively.}~\cite{lovasz, alon_null}. We illustrate one such perfect matching in \eqref{eq:18_eq}, where the colored unknown symbol in a combination represents the unknown symbol matched by that combination. The similar argument can be performed for the remaining combinations of $2$ failed systematic nodes.

\end{example}

\subsubsection{MSMR point}

The achievability results above together with the repair bandwidth bound reported in the previous section, see Remark~\ref{rem:MSMR}, results in the following characterization.

\begin{theorem}
The MSMR point for the $(n,k,d,t,\alpha,\gamma)$-CMR model is given by
$$\alpha_{MSMR}=\frac{\Mc}{k}, \quad \gamma_{MSMR}=\frac{\Mc d t}{k(d-k+t)}.$$
\end{theorem}

\section{MBMR Codes}
\label{sec:mbmr_codes}

In this section, we focus on the other extremal point of the storage vs. repair-bandwidth trade-off, namely the MBMR point.

\subsection{Repair bandwidth bound}

For the MBMR point, depending on whether $t | k$ or $t \nmid k$, we state the following two results. 

\begin{proposition}
\label{prop:MBMR}
Assume that $t | k$. Consider an $(n, k)$-DSS that stores a file of size $\Mc$ 
and enables repair of $t$ failed nodes under a $(d, t, \alpha_{MBMR}, \gamma_{MBMR})$-CMR model. Then, denoting the entropy of $t$ nodes as $H_t$, we have 
\begin{align}
&t\alpha_{MBMR}\geq H_t= \gamma_{MBMR}, \label{eq:MBMRa}\\
&\gamma_{MBMR}\geq \frac{\Mc2dt}{k(2d - k + t)}.\label{eq:MBMRb}
\end{align}
\end{proposition}

\begin{proof}
Note that the MBMR point has $H(\xv_{\Sc}) = \gamma_{MBMR}$ for every $\Sc \subseteq [n]$ such that $|\Sc| = t$. Therefore, we have
\begin{align}
\gamma_{MBMR} = H(\xv_{\Sc}) \leq \sum_{i \in \Sc}H(\xv_i) \leq t\alpha_{MBMR}. \nonumber
\end{align}
In order to establish the lower bound on $\gamma_{MBMR}$ in \eqref{eq:MBMRb}, we use $n_i=t, \forall i \in [a]$ in the bound \eqref{eq:bound}, we obtain
\begin{align}
\label{eq:MBMR_1}
\Mc\leq \sum\limits_{i=1}^{k/t}  \big(d-(i-1)t\big)\beta=\frac{k}{t}\left(\frac{2d-k+t}{2}\right)\beta.
\end{align}
This implies that 
$\gamma_{MBMR}=d\beta\geq \frac{\Mc 2dt}{k(2d-k+t)}.$
\end{proof}

\begin{proposition}
\label{prop:MBMR2}
Consider an $(n, k)$-DSS that stores a file of size $\Mc$ 
and enables repair of $t$ failed nodes under a $(d, t, \alpha_{MBMR}, \gamma_{MBMR})$-CMR model. Then, the bounds given in \eqref{eq:MBMRa} and \eqref{eq:MBMRb} hold for the case of $t \nmid k$, if
$H_b \geq \left(\frac{\beta}{t}\right) \left[b\left(\frac{2d + t - 1}{2}\right) - {b \choose 2} \right]$,
where $b=k~(\textrm{mod}~t)$, and $H_b$ denotes entropy of $b$ nodes in the system.
\end{proposition}

\begin{proof} 
The bound in \eqref{eq:MBMRa} follows from the similar analysis as presented in the proof of Proposition~\ref{prop:MBMR2}. In order to establish \eqref{eq:MBMRb}, we select $g = \floorb{k/t} + 1 = a+1$ disjoint sets of nodes indexed by the sets $\Sc_1, \Sc_2,\ldots, \Sc_{g}$ such that $n_1 = |\Sc_1| = b$ and $n_i = |\Sc_i| = t$ for $i \in \{2, 3,\cdots, g = a+1\}$. Note that we have $\sum_{i}n_i = k$. Utilizing this particular sequence of sets in \eqref{eq:bound_h} along with the fact that we have $H(\xv_{\Sc_i}) = d\beta$ for $2 \leq i \leq g$, we obtain
\begin{align}
\label{eq:tmp_step11}
\Mc & \leq  \min\big\{ H(\xv_{\Sc_1}), d\beta\big\} + \sum\limits_{i=1}^{a} \big(d-(i-1)t - b\big)\beta \nonumber \\
& =  H(\xv_{\Sc_1}) + \sum\limits_{i=1}^{a} \big(d-(i-1)t - b\big)\beta.
\end{align}
Note that the choice of the set $\Sc_1$ is arbitrary and all the nodes in the system are equivalent in terms of their information content. Therefore, $H_b = H(\Sc_1)$ (the amount of information stored on $b$ nodes indexed by the set $\Sc_1$) only depends on $b$. It follows from \eqref{eq:tmp_step11} that 
\begin{align}
\label{eq:MBMR_gen}
\Mc &\leq 
H_b +\left(\frac{2d - k + (t - b)}{2}\right)a\beta 
\end{align}

In order to have the bound in \eqref{eq:MBMRb} 
we need the RHS of \eqref{eq:MBMR_gen} to be at least the RHS of \eqref{eq:MBMR_1}, i.e.,
\begin{align}
H_b +\left(\frac{2d - k + (t - b)}{2}\right)a\beta  \geq \frac{k}{t}\left(\frac{2d-k+t}{2}\right)\beta. \nonumber
\end{align} 
This implies that 
\begin{align}
H_b &\geq \left(\frac{2d-k+t}{2}\right)\frac{k}{t}\beta - \left(\frac{2d - k + (t - b)}{2}\right)a\beta \nonumber \\
& = \left(\frac{\beta}{t}\right) \left[b\left(\frac{2d + t - 1}{2}\right) - {b \choose 2} \right].
\end{align}
\end{proof}

\begin{remark}\label{rem:MBMR}
A code that allows for repair of $t$ failed nodes with $H_t = \gamma  = \frac{\Mc 2 d t}{k(2d - k + t)}$ is an MBMR code for the case of $t | k$ and $t\nmid k$, if for the latter case the system also operates at
$H_b \geq \left(\frac{\beta}{t}\right) \left[b\left(\frac{2d + t - 1}{2}\right) - {b \choose 2} \right]$.
\end{remark}

\subsection{Constructions and the characterization of the MBMR point}
\label{sec:mbmr_construction}

\subsubsection{Constructions from existing MBCR codes}
MBCR codes have $\alpha_{MBCR}=\frac{\Mc}{k}\frac{2d+t-1}{2d+t-k}$, $\beta=\frac{\Mc}{k}\frac{2}{2d+t-k}$, and $\beta'=\frac{\Mc}{k}\frac{1}{2d+t-k}$. A construction of MBCR codes for all parameters is provided in \cite{WangZhang13}, where the entropy accumulation for MBCR codes is also characterized. In particular, entropy of $b\leq k$ nodes is given by $H_b=\left(b\left(\frac{2d+t-1}{2}\right)- {b \choose 2} \right)\beta$. 

\begin{proposition}
\label{prop:mbmr_mbcr}
A code $\Cc$ that operates as an MBCR code is also an MBMR code for the CMR model that operates at 
$\alpha=\frac{\Mc (2d+t-1)}{k(2d+t-k)}$ and 
$H_b \geq \left(\frac{\beta}{t}\right) \left[b\left(\frac{2d + t - 1}{2}\right) - {b \choose 2} \right]$.
\end{proposition}

\begin{proof}
Consider that each failed node contact to the same set of $d$ nodes in the MBCR code $\Cc$. This results in a repair  bandwidth of at most $\gamma=td\beta_{MBCR}=\frac{\Mc 2dt}{k(2d+t-k)}$. Entropy of $t$ nodes in this code is given by $H_t=\left(t\left(\frac{2d+t-1}{2}\right)- {t \choose 2} \right) \frac{\Mc}{k}\frac{2}{2d+t-k}=\frac{\Mc 2dt}{k(2d+t-k)}=\gamma$. These and also the entropy of $b$ nodes meet the conditions stated in Remark~\ref{rem:MBMR}, establishing the claimed result.
\end{proof}

\begin{remark}
In general, for MBMR codes, we have the condition that $t\alpha\geq H_t = \gamma_{MBMR}$. It is not clear if $\alpha$ can be further reduced than that in Proposition~\ref{prop:mbmr_mbcr}, e.g., when $b = 0$.
\end{remark}

\subsubsection{MBMR point}

The achievability results above together with the repair bandwidth bound reported in the previous section results in the following characterization.

\begin{theorem}
Let $k~({\rm mod}~t) = b$. Then,  for the CMR models satisfying
$H_b \geq \left(\frac{\beta}{t}\right) \left[b\left(\frac{2d + t - 1}{2}\right) - {b \choose 2} \right]$, the MBMR point is given by
$$H_t=\gamma_{MBMR}=\frac{\Mc 2dt}{k(2d+t-k)}.$$
\end{theorem}

\section{Applications to communication and repair efficient secret sharing schemes}
\label{sec:comm_secret}
 
Recently, in \cite{HLKB15}, Huang et al. proposed a model for {\em communication efficient secret sharing}. They consider a setting where one wants to encode a secret $\mv \in \ff_Q^K$ into $N$ shares $\sv_1, \sv_2,\ldots, \sv_N \in \ff_Q$. The encoding from secret to shares should satisfy two requirements: 1) Given any $z$ shares one should not be able to learn any information about the secret $\mv$ and 2) given access to any $d \geq N - r$ shares one should be able to reconstruct (or decode) the entire secret $\mv$. Huang et al. refer to such secret sharing schemes as $(N, K, r, z)_Q$ secret sharing schemes. For a naive secret reconstruction process, one downloads $Q$ symbols over $\ff_Q$ from each of the $d$ contacted shares leading to the communication bandwidth (the amount of data downloaded for secret reconstruction) of $dQ$ symbols over $\ff_Q$. In \cite{HLKB15}, Huang et al. explore the minimum possible communication bandwidth of an $(N, K, r, z)_Q$ secret sharing scheme as a function of the number of shares participating in the reconstruction process $d$. Towards this end, the authors obtain the following bound on the communication bandwidth of an $(N, K, r, z)_Q$ secret sharing scheme when $N - r \leq d \leq N$ shares are available during the reconstruction process\footnote{In \cite{HLKB15}, the authors present this bound in terms of communication overhead $CO_d$ which is the difference between the communication bandwidth $BW_d$ and the size of the secret $K$.}. 
\begin{align}
\label{eq:bw_bound}
BW_{d} \geq \frac{d}{d - z}K,
\end{align}
where the communication bandwidth is counted in terms of the number of symbols over $\ff_Q$. Huang et al. further present an explicit $(N, K = N - r - z, r, z)_Q$ secret sharing scheme which attain the bound in \eqref{eq:bw_bound} for $d = k$ and $d = N$. They also show the existence of $(N, K = N - r - z, r, z)_Q$ secret sharing schemes which attain the lower bound on the communication bandwidth for all values of $d$ in $\{N - r, N - r + 1,\ldots, N\}$. Note that these secret sharing schemes are designed to work for a particular value $d$. However, it is also an interesting question to design secret sharing schemes which simultaneously work for all the values of $d$. In \cite{BR15}, Bitar and El Rouayheb present two explicit constructions which give $(N, K = N - r - z, r, z)_Q$ secret sharing schemes with optimal communication bandwidth. The first construction attains the bound in \eqref{eq:bw_bound} for any fixed $d$ and the second construction simultaneously attains the bound for $N - r \leq d \leq N$.

In this section, we show that the communication optimal $(N, K, r, z)_Q$ secret sharing schemes can be designed using the codes which allow for centralized repair of multiple nodes. The added advantage of using this approach to construct secret sharing scheme is that this method also enables bandwidth efficient repair of shares in the secret sharing scheme. This can also be viewed as an attempt to unify the study of repair bandwidth efficient codes for distributed storage and communication efficient efficient secret sharing. This allows us to employ various ideas from the work on secure distributed storage literature to the setting of communication efficient secret sharing.

Let $\Mc^{s}$ be the size of the secret $\mv$ (over $\ff_{q}$) that we want secure in the secret sharing scheme. We further assume that each of the $N$ shares in the secret sharing scheme consists of $\alpha$ symbols over $\ff_q$, i.e., we have $\ff_Q = \ff_{q^{\alpha}}$. Note that Huang et al. define the sizes of the secret and the shares over the same alphabet $\ff_Q = \ff_{q^{\alpha}}$~\cite{HLKB15}. However, we denote the size of the secret over a base field $\ff_q$ and assume that each share comprises a symbol from the extension field $\ff_Q = \ff_{q^{\alpha}}$. This representation is quite prevalent in the distributed storage literature and is consistent with the rest of the paper as well. We represent the secret sharing scheme as an $(N, \Mc^{s}, r, z)_{\alpha, q}$ or $(N, \Mc^{s}, r, z)_{\alpha}$ secret sharing scheme. First, we restate the lower bound on the communication bandwidth of an $(N, \Mc^{s}, r, z)_{\alpha}$ secret sharing schemes (cf.~\eqref{eq:bw_bound}) in our notations as follows. 
\begin{align}
\label{eq:bw_bound1}
BW_{d} \geq \frac{d}{d - z}\Mc^{s},
\end{align}
where we count the communication bandwidth $BW_{d}$ in terms of number of symbols over the base field $\ff_q$.

\begin{definition}{($z$-secure distributed storage system)} Consider an $(n, k)$-DSS storing a file $\fv^s$ of size $\Mc^{s}$ (over $\ff_q$) under the $(d, t)$-CMR model. We say that the DSS is {\em $z$-secure} if an eavesdropper who has access to the content of any set of $z$ (out of $n$) storage nodes does not gain any information about the file $\fv^{s}$. 
\end{definition}

\begin{remark}
Recall that when there is no security requirement, we denote the file stored on the DSS and its size as $\fv$ and $\Mc$ (over $\ff_q$), respectively (cf.~Section~\ref{sec:gen}). The quantity $\Mc - \Mc^{s}$ denotes the loss in the file size that the system has to bear in order to guarantee the information theoretic security of the stored file against an eavesdropper. Or, this part of the data can be considered as public information (without any secrecy constraints), as compared to the private counterpart (which has secrecy constraints).
\end{remark}

The file size bounds for DSS which are secure against even a general eavesdropping model where an eavesdropper can observe both the content stored on a set of nodes and the content downloaded during the repair of another set of node have been previously considered in the literature. The regenerating coding schemes which are secure against such eavesdroppers are presented in \cite{PRR11, SRK_globecom11, RKSV12, HPX15a} and references therein. Similarly, the problem of designing secure cooperative regenerating codes is explored in \cite{KRV12, HPX15b}. As discussed in Section~\ref{sec:msmr_codes} and \ref{sec:mbmr_codes}, both regenerating codes and cooperative regenerating codes are specific sub-classes of codes for centralized repair model. Therefore, both the secure regenerating codes and secure cooperative regenerating codes which can prevent the leakage of information to an eavesdropper observing the content stored on $z$-storage nodes form special cases of $z$-secure DSS under CMR  model with respective system parameters. 

We can utilize $z$-secure coding scheme for DSS under the CMR model to obtain communication efficient secret sharing schemes. We first illustrate this approach with the help of a secure MSR code in the following subsection. We then comment on how this approach can be employed using general secure coding schemes for DSS under the CMR model.

\subsection{An example}
Let $\Cc$ be a linear systematic code which operates at $(n, \Mc, d < n - 1, \alpha = \frac{\Mc}{z + 1}, \beta = \frac{\alpha}{d - (z+1) + 1})_q$ MSR point. Note that this $\Cc$ is also an MDS code where the content of any $k = z+1$ symbols is sufficient to recover the entire file of size $\Mc$. We next show how we can use $\Cc$ to construct a communication bandwidth efficient $(N = n - 1, \Mc^{s} = \alpha = \frac{\Mc}{z + 1}, r = N - z - 1, z)_{\alpha}$ secret sharing scheme. 

Let $\mv = (m_1,\ldots, m_{\alpha}) \in \ff_q^{\alpha}$ denote the secret of size $\alpha$ over $\ff_q$. Let $\rv = (r_1, r_2,\ldots, r_{z\alpha}) \in \ff_q^{z\alpha}$ be $z\alpha$ random symbols which are distributed uniformly at random over $\ff_q$. We encode the $\Mc = (z + 1)\alpha$-length vector $(\mv, \rv) \in \ff_q^{\Mc}$ using the MSR code $\Cc$. Let $(\cv_1, \cv_2,\ldots, \cv_n) = (\mv, \rv, \cv_{z+2}, \ldots, \cv_n) \in \ff_{q^{\alpha}}^{n}$ denote the associated MSR codeword. Note that a code symbol, say $\cv_i$, can be repaired by any set of $d$ out of the remaining $n-1$ code symbols by downloading at most $d\beta = \frac{d}{d - (z+1) + 1}\alpha = \frac{d}{d - z}\Mc^{s}$ symbols (over $\ff_q$) from the contacted $d$ nodes. In order to obtain a secret sharing scheme we puncture the symbol $\cv_1$ from each of the codewords in $\Cc$ which gives us another code $\tilde{\Cc} \in \ff_{q^{\alpha}}^{n-1}$. Let $\tilde{\cv} = (\cv_2, \cv_3,\ldots, \cv_{n}) \in \ff_{q^{\alpha}}^{n-1}$ be the codeword in $\tilde{\Cc}$ which is obtained by removing the first code symbol from the codeword $\cv \in \Cc$.  For the secret $\mv$ we treat $n-1$ symbols in $\tilde{\cv}$ as $N = n - 1$ shares of the secret sharing scheme. In order to reconstruct the secret $\mv$, we can invoke the node repair process of first node (code symbol) in the original MSR code $\Cc$ where we contact $d$ shares and download $\beta =\frac{\alpha}{d-k+1}$ from each of these $d$ shares. This leads to the communication bandwidth of 
$$
d\beta =  \frac{d}{d - z}\Mc^{s},
$$
which matches the bound in \eqref{eq:bw_bound1}. Using the MDS property \footnote{Since $\Cc$ is an $(n, z+1)$ MDS code, it is straightforward to observe that $\tilde{\Cc}$ is an $(n - 1, z+1)$ MDS code.} of $\tilde{\Cc}$, it is easy to argue that $\Cc$ is a $z$-secure coding scheme. Note that besides reconstructing the secret $\mv$ in a communication efficient manner, the proposed scheme also allows the bandwidth efficient repair of any of the $N = n - 1$ shares by using $d$ out of $N-1 = n-2$ remaining shares. This can be performed again by invoking the repair mechanism of the original MSR code $\Cc$.

\begin{figure}[htbp]
	\centering
		\includegraphics[width=0.25\textwidth]{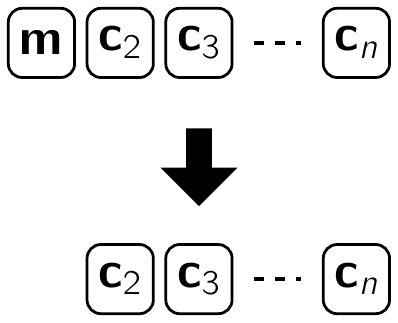}
          \caption{Message puncturing from an MSR code gives communication efficient secret sharing with repairable shares.}
	\label{fig:msr_example}
\end{figure}

\subsection{Construction of communication and repair efficient secret sharing schemes using MSMR codes}

Generally, we can utilize an MSMR code to obtain a communication efficient secret sharing scheme which also enables bandwidth efficient repair of the shares in the scheme. Let $\Cc$ be a systematic linear $(n, k = z + t, d, t, \alpha = \frac{\Mc}{k}, \gamma = \frac{t\Mc}{k}\frac{d}{d - k + t})_q$-MSMR code. Recall that this code encodes a file $\fv$ of size $\Mc$ over $\ff_q$ to an $n$-length codeword $\cv = (\cv_1,\cv_2,\ldots, \cv_n) \in \ff_{q^{\alpha}}^n$ such that we have
$$
\cv_i = (\fv_{(i-1)\alpha + 1}, \fv_{(i-1)\alpha + 2},\ldots, \fv_{i\alpha}) \in \ff_q^{\alpha}~\text{for}~1 \leq i \leq k.
$$
Using the code $\Cc$, we now construct a communication efficient $(N = n - t, \Mc^{s} = t\alpha = \frac{t\Mc}{z+t}, r = N - z - t= n - k - t, z)_{\alpha}$ secret sharing scheme. Let $\mv \in \ff_q^{\Mc^{s}} = \ff_q^{t\alpha}$ denote the secret to be encoded. Let $\rv = (r_1, r_2,\ldots, r_{z\alpha}) \in \ff_q^{z\alpha}$ be $z\alpha$ independent random symbols which are distributed uniformly at random over $\ff_q$. We encode the $\Mc = (z + t)\alpha$ symbols long file $\fv = (\mv, \rv) \in \ff_q^{(z+t)\alpha}$ using the MSMR code $\Cc$. Given $\cv = (\cv_1, \cv_2,\ldots, \cv_n) \in \ff^n_{q^{\alpha}}$, the codeword associated with the file $\fv$ in the MSMR code $\Cc$, we puncture the codeword at the first $t$ code symbols to obtain a punctured codeword $\tilde{\cv} = (\tilde{\cv}_1, \tilde{\cv}_2,\ldots, \tilde{\cv}_{N}) = (\cv_{t+1}, \cv_{t+2},\ldots, \cv_{n}) \in \ff_{q^{\alpha}}^{n-t}$. Assuming that $\tilde{\Cc}$ denotes the codebook obtained by puncturing all the codewords in $\Cc$ at the first $t$ code symbols, we have $\tilde{\cv} \in \tilde{\Cc}$. We claim that $\tilde{\Cc}$ gives us a $(N = n - t, \Mc^{s} = t\alpha = \frac{t\Mc}{z+t}, r = N - z = n - k, z)_{\alpha}$ secret sharing scheme.

\begin{figure}[htbp]
	\centering
		\includegraphics[width=0.50\textwidth]{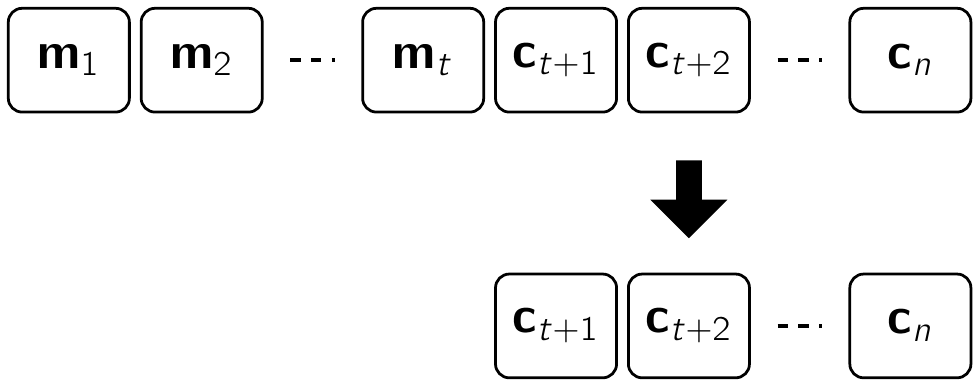}
          \caption{Puncturing of the secret covering multiple symbols in an MSMR code gives communication efficient secret sharing with multi-share repair property.}
	\label{fig:mscr_example}
\end{figure}

\begin{itemize}
\item \textbf{Security:}~Let's consider an adversary who has access to $z$ shares $\tilde{\cv}_{i_1},\tilde{\cv}_{i_2},\ldots, \tilde{\cv}_{i_z}$. We make two observations: First, we have $H(\tilde{\cv}_{i_1},\tilde{\cv}_{i_2},\ldots, \tilde{\cv}_{i_z}) \leq z\alpha =  H(\rv)$. Second, given $\mv$ and $\tilde{\cv}_{i_1},\tilde{\cv}_{i_2},\ldots, \tilde{\cv}_{i_z}$ we have access to $k = z + t$ code symbols of $\cv \in \Cc$; as a result, we can decode $\rv$ (by decoding $\fv = (\mv, \rv)$) as $\Cc$ is an $(n, k = z+ t)$ MDS code. From these two observations, it follows that the adversary does not get any information about the secret $\mv$ from the $z$ shares at its disposal~\cite{SRK_globecom11, RKSV12}.
\item \textbf{Communication efficiency:~}Assume that we contact a set of $d$ shares $\tilde{\cv}_{i_1},\tilde{\cv}_{i_2},\ldots, \tilde{\cv}_{i_d}$. Since these $d$ shares form $d$ code symbols in the codeword $\cv$ of the MSMR code, we can use these $d$ shares to recover the secret $\mv$ which constitutes the first $t$ code symbols of the code word $\cv$. Recall that we download total $\gamma = \frac{t\Mc}{z+t}\frac{d}{d - (z+t) +t} = \frac{d}{d - z}{\Mc^{s}}$ symbols (over $\ff_q$) from the $d$ shares we contact. Note that $\gamma$ is exactly equal to the lower bound on \eqref{eq:bw_bound1} which establishes the communication efficiency of the obtained secret sharing scheme.
\end{itemize}

\begin{remark}
Note that if we have $d < N - t$, by invoking the repair process of the original MSMR code $\Cc$, we can repair any $t$ shares by contacting any set of $d$ out of $N-t$ remaining shares and downloading $\gamma = \frac{d}{d - z}{\Mc^{s}}$ symbols from the contacted shares. We are not necessarily required to repair the shares in the group of $t$ failed shares at a time. If the original MSMR coding scheme also allows for bandwidth efficient repair of less than $t$ code symbols (nodes) at a time, then we can also repair less than $t$ failed shares at a time by using such repair mechanism. Specifically, in Section~\ref{sec:msmr1} and \ref{sec:msmr2}, we discuss some constructions of MSMR codes that are obtained from MSR codes. Therefore, the secret sharing schemes designed by these codes enable bandwidth efficient repair of one share at a time as well.
\end{remark}

\subsection{Secret sharing schemes using MBMR codes}
\label{sec:mbmr_ss}

In this subsection, we illustrate how coding schemes at the MSMR point can be utilized to construct communication and repair efficient secret sharing scheme. Note that this allows us to increase the size of the share in an MDS code in order to lower the repair bandwidth. Furthermore, this also allows us to construct explicit secret sharing schemes for a wider set of parameters $n, k, d$, and $\alpha$. Here we note that the MBMR coding scheme that we employ in this subsection is from \cite{WangZhang13}. We define the following quantities. 
\begin{align}
\Mc &= k(2d + t - k) = (z + t)(2d + t - (z + t)) = (z + t)(2d - z) \\
\Mc^{s} &= t(2d + t - k - z) = t(2d + t - (z + t) - z) = 2t(d - z). 
\end{align}
Given $n, k = z+t$ and  $d$, we construct an MBMR code as follows. 
\begin{itemize}
\item Let $\{y_1, y_2,\ldots, y_{n + d + t - 1}\} \subseteq \ff_q$ be $n + d + t -1$ distinct elements in $\ff_q$. Similarly, we select another set of $n + d -1$ distinct elements in $\ff_q$ as $\{x_1, x_2,\ldots, x_{n + d -1}\} \subseteq \ff_q$.
\item Given an $\Mc$-length message vector $\fv = (f_1, f_2,\ldots, f_{\Mc}) \in \ff^{\Mc}_q$ construct a bi-variate polynomial such that 
\begin{align}
\label{eq:mbmr_poly}
F(X, Y) = \sum_{\substack{0 \leq i < k,\\ 0 \leq j < k}}a_{i, j}X^iY^j +  \sum_{\substack{0 \leq i < k,\\ k \leq j < d+t}}b_{i, j}X^iY^j  +  \sum_{\substack{k \leq i < d,\\ 0 \leq j < k}}c_{i, j}X^iY^j.
\end{align}
Here,
\begin{align}
\label{eq:f_precode}
(a_{0, 1}, a_{0, 2},\ldots, a_{k-1, k-1}, b_{0, k}, b_{1, k},\ldots, b_{k-1,d+t-1}, c_{k, 0}, c_{k,1},\ldots, c_{d-1,k-1}) = \Am\fv
\end{align}
for an $\Mc \times \Mc$ matrix $\Am$ with entries from $\ff_q$ which we specify later.
\item Given the polynomial $F(X, Y)$, the $i$th code symbol $\cv_i$ of the codeword associated with $\fv$ in $\Cc$ is obtained by evaluating $F(X, Y)$ at 
$$
\{(x_i, y_i), (x_{i}, y_{i+1}),\ldots, (x_i, y_{i + d + t -1}), (x_{i+1}, y_i), (x_{i+2}, y_i),\ldots, (x_{i+d-1}, y_i)\}.
$$
That is, we have
$$
\cv_i = \big(F(x_i, y_i), F(x_{i}, y_{i+1}),\ldots, F(x_i, y_{i + d + t -1}), F(x_{i+1}, y_i), F(x_{i+2}, y_i),\ldots, F(x_{i+d-1}, y_i)\big) \in \ff_q^{2d + t -1}.
$$
\end{itemize}
\begin{remark}
\label{rem:wang1}
Note that the code symbol $\cv_i$ contains $d+t$ evaluations of the degree-$(d+t)$ polynomial $h_i(Y) =  F(x_i, Y)$ at distinct points $\{y_i, y_{i+1},\ldots, y_{i+d+t-1}\}$. Therefore, the content of $\cv_i$ is sufficient to recover the polynomial $h_i(Y) =  F(x_i, Y)$. Similarly, $\cv_i$ contains $d$ evaluations of the degree-$d$ polynomial $g_i(X) =  F(X, y_i)$ at distinct points $\{x_i, x_{i+1},\ldots, x_{i+d -1}\}$. This implies that the content of $\cv_i$ is sufficient to recover the polynomial $g_i(X) = F(X, y_i)$.
\end{remark}

\begin{remark}
\label{rem:wang2}
In \cite{WangZhang13}, Wang and Zhang show that this construction enables repair of any $t$ code symbols (nodes) under a cooperative repair framework. This implies that the coding scheme can also be utilized in the centralized repair framework. As discussed in Section~\ref{sec:mbmr_construction}, these codes operate at the MBMR point with $\alpha = 2d + t - 1$ and the repair bandwidth 
\begin{align}
\label{eq:wang_bw}
\gamma_t = \frac{\Mc2dt}{k(2d + t - k)}  = \frac{\Mc2dt}{(z + t)(2d  - z)}.
\end{align}
\end{remark}

The codeword associated with the information symbols $\fv$ in $\Cc$ is described in Figure~\ref{fig:wang1}. Note that the content of evaluations highlighted in Figure~\ref{fig:wang1} form an information set as the original $\Mc$ information symbols $\fv$ can be reconstructed from the highlighted symbols~\cite{WangZhang13}. Therefore, it is possible to precode the information symbols  $\fv$ using an $\Mc \times \Mc$ matrix $\Am$ (cf.~\eqref{eq:f_precode}) such that the information symbols themselves appear at the highlighted positions in the codewords of $\Cc$. Note that this corresponds to a systematic encoding for the code $\Cc$.

\begin{figure}[htbp]
	\centering
		\includegraphics[width=1.0\textwidth]{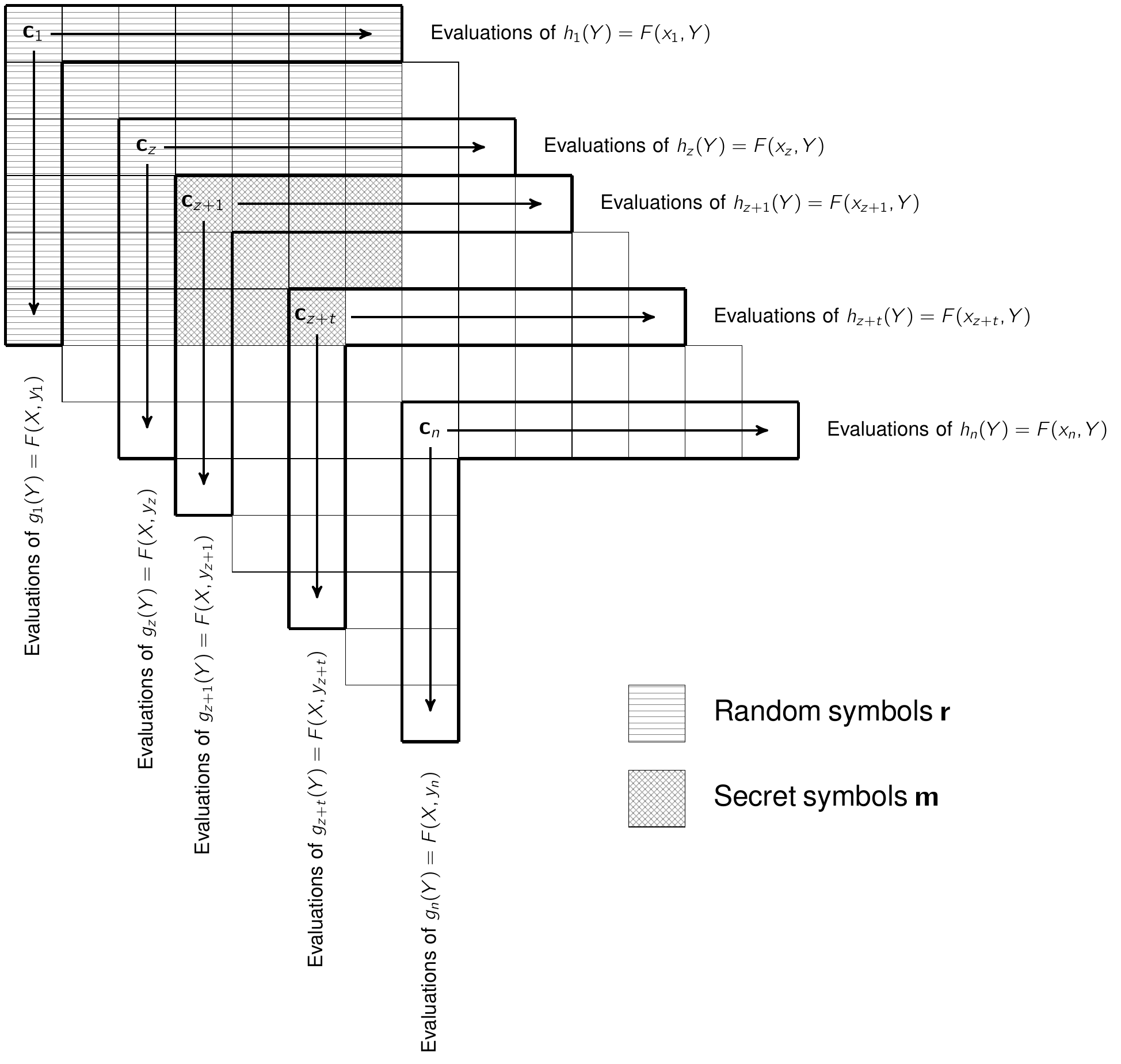}
          \caption{Description of the MBMR coding scheme with a particular systematic encoding utilized in this paper. Each node (code symbol) is collection of $d+t$ and $d$ evaluations of the univariate polynomials $\{h_i(Y)\}$ and $\{g_i(X)\}$, respectively. Note that the systematic encoding ensures that the random symbols $\rv$ and the secret symbols $\mv$ appear in the first $z$ and the subsequent $t$ code symbols, respectively.}
	\label{fig:wang1}
\end{figure}

We now describe how we can utlize the MBMR code described above to obtain a a communication efficient $(N = n - t, \Mc^{s}, d, r = N - z - t = n - z - 2t, z)_{\alpha, q}$-secret sharing scheme. Let $\mv = (m_1, m_2,\ldots, m_{\Mc^{s}}) \in \ff^{\Mc^{s}}_q$ be a $\Mc^{s}$-length (over $\ff_q$) secret that needs to be encoded in the secret sharing scheme. Let $\rv = (r_1, r_2,\ldots, r_{\Rc}) \in \ff_q^{\Rc}$ be $\Rc = \Mc - \Mc^{s} = 2dz + zt -z^2$ i.i.d. random variables which are uniformly distributed over $\ff_q$. Given the $\Mc^{s}$-length secret $\mv$, the $(\Mc - \Mc^{s})$-length random symbols $\rv$ and the precoding matrix $\Am$, we construct the $\Mc$-length information vector \footnote{Each of the $\Mc$ symbols in the vector $\fv$ comprises either a symbol from $\mv$ or $\rv$. Moreover, each symbol of $\mv$ and $\rv$ appears in exactly one coordinate of the vector $\fv$.} $\fv$ such that the secret $\mv$ appears in the code symbols $\cv_{z+1},\ldots, \cv_{z + t}$ as described in Figure~\ref{fig:wang1}. In other words, there exists a permutation $\sigma: [\Mc^{s}] \rightarrow [\Mc^{s}]$ such that we have
\begin{align}
\label{eq:secret_poly}
F(x_{z +1}, y_{z+1}) = m_{\sigma(1)}, F(x_{z+1}, y_{z+2}) &= m_{\sigma(2)},\ldots, F(x_{z+1}, y_{z + d + t - 1}) = m_{\sigma(d+t)}, \nonumber \\
&\vdots \nonumber \\
F(x_{z + t}, y_{z+1}) = m_{\sigma((t-1)(d+t) + 1)}, F(x_{z + t}, y_{z+2}) &= m_{\sigma((t-1)(d+t) + 2)},\ldots, F(x_{z + t}, y_{z + d + t - 1}) = m_{\sigma(t(d+t))} \nonumber \\
F(x_{z + t+1}, y_{z+1}) = m_{\sigma(t(d+t)+1)}, F(x_{z + t+1}, y_{z+2}) &= m_{\sigma(t(d+t)+2)},\ldots, F(x_{z + t+1}, y_{z + t}) = m_{\sigma(t(d+t)+t)}, \nonumber \\
&\vdots \nonumber \\
F(x_{d}, y_{z+1}) = m_{\sigma((t-1)(d+t) + 1)}, F(x_{d}, y_{z+2}) &= m_{\sigma((t-1)(d+t) + 2)},\ldots, F(x_{d}, y_{z + t})  = m_{\sigma(\Mc^{s})},
\end{align}
where we have used the fact that $\Mc^{s} = 2t(d - z)$ in the last equality. Now, the encoding of the secret $\mv$ in the secret sharing scheme $\tilde{\Cc}$ is obtained by puncturing the code symbols $\cv_{z+1}, \cv_{z+2},\ldots, \cv_{z+t}$ for the codeword $\cv \in \Cc$ (cf. Figure~\ref{fig:wang2}). Thus, the $N = n - t$ shares associated with the secret $\mv$ in the secret sharing scheme $\tilde{\Cc}$ are defined as 
\begin{align}
\tilde{\cv} = (\tilde{\cv}_1, \tilde{\cv}_2,\ldots, \tilde{\cv}_N) = ({\cv}_1,\ldots, \cv_{z}, \cv_{z + t + 1},\ldots, \cv_{n}) \in \ff_{q^{\alpha}}^{N}.
\end{align}

\begin{figure}[t!]
	\centering
		\includegraphics[width=1.0\textwidth]{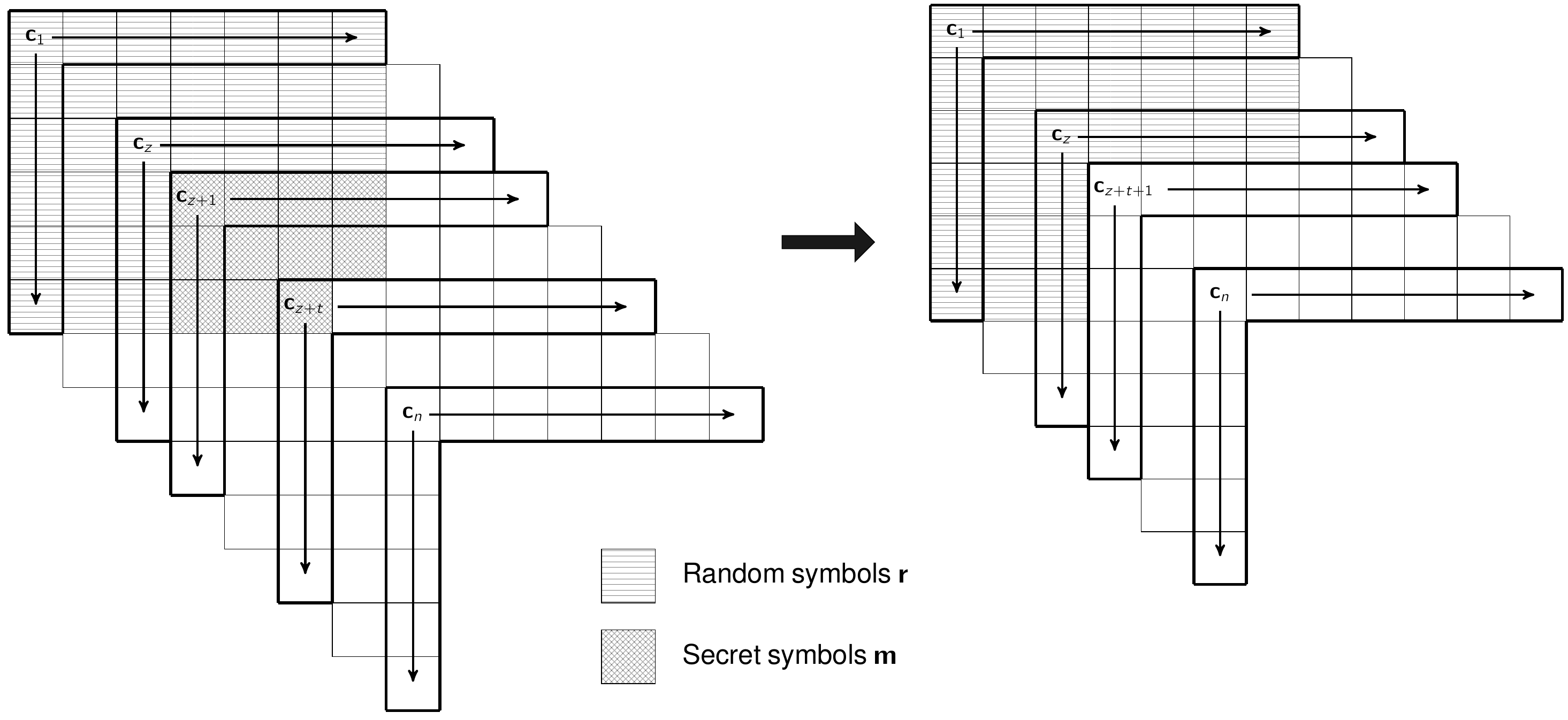}
          \caption{Puncturing of the secret covering multiple symbols in the MBMR code gives communication efficient secret sharing with multi-share repair property.}
	\label{fig:wang2}
\end{figure}

We now argue the security and the communication efficiency of the proposed secret sharing scheme.
\begin{itemize}
\item \textbf{Security:}~Assume that an adversary has access to $z$ shares $\tilde{\cv}_{i_1},\tilde{\cv}_{i_2},\ldots, \tilde{\cv}_{i_z}$. Let these shares correspond to the code symbols $\cv_{j_1}, \cv_{j_2},\ldots, \cv_{j_z}$ in the associated codeword in the MBMR code $\Cc$ with $\{j_1, j_2,\ldots, j_s\} \subset [n]\backslash \{z+1,\ldots, z+t\}$. It follows from the Remark~\ref{rem:wang1} that the adversary knows the following univariate polynomials\footnote{Knowing a polynomial  means that the adversary knows the coefficients of the polynomials and can evaluate the polynomial at any point in $\ff_q$}. 
\begin{align}
\label{eq:adversaryMBMR}
\big\{h_{j_s}(Y) = F(x_{j_s}, Y), g_{j_s}(X) = F(X,y_{j_s})\big\}~~\text{for}~s = 1, 2,\ldots, z.
\end{align}
It is argued in \cite{KRV12} that if $\ev$ denote the symbols (in $\ff_q$) known to the adversary by observing $z$ shares, then we have
\begin{align}
H(\ev) \leq H(\rv) = \Mc - \Mc^{s}.
\end{align}
Next, we argue that given the observations of the adversary (cf.~\eqref{eq:adversaryMBMR}) and the secret $\mv$, one can decode the random symbols $\rv$, i.e., $H(\rv|\mv, \ev) = 0$. Note that the secret symbols $\mv$ correspond to part of the code symbols $\cv_{z+1},\ldots, \cv_{z+t}$. As highlighted in \eqref{eq:secret_poly}, the secret symbols $\mv$ are evaluations of the polynomial $F(X, Y)$ (cf.~\eqref{eq:mbmr_poly}). In particular, knowing the secret $\mv$ translates to knowing $d + t - z$ evaluations of each of the polynomials in $\{h_j(Y) = F(x_j, Y)\}_{j = z+1,\ldots, z+t}$ and $d - z$ evaluations of each of the polynomials in $\{g_j(X) = F(X,y_{j})\}_{j = z+1,\ldots, z+t}$ (cf.~\eqref{eq:secret_poly}). Now, using the observations of the adversary (cf.~\eqref{eq:adversaryMBMR}), we can obtain $z$ additional observations of each of the polynomials $\{h_j(Y) = F(x_j, Y), g_j(X) = F(X,y_{j})\}_{j = z+1,\ldots, z+t}$ as follows.
\begin{align}
h_j(y_{j_s}) = F(x_j, y_{j_s}) = g_{j_s}(x_j)~\text{for}~j \in \{z+1,\ldots, z+t\}~\text{and}~s\in \{1,\ldots, z\}
\end{align}
and
\begin{align}
g_j(x_{j_s}) = F(x_{j_s}, y_j) = h_{j_s}(y_j)~\text{for}~j \in \{z+1,\ldots, z+t\}~\text{and}~s\in \{1,\ldots, z\}.
\end{align}
Therefore, given the observations of the adversary and the secret symbols one has access to $k = z+t$ code symbols $\cv_{j_1},\ldots, \cv_{j_s}, \cv_{z+1},\ldots, \cv_{z+t}$ of the associated codeword in $\Cc$. Now one can use a decoding algorithm of $\Cc$ to decode $\Am\fv$ and subsequently obtain $\fv$. Note that the random symbols $\rv$ can now be obtained as these symbols constitute $\Mc - \Mc^{s}$ coordinates of the vector $\fv$. 

From the two observations shown above that $H(\ev) \leq H(\rv)$ and $H(\rv|\ev, \mv) = 0$, it follows that the adversary does not get any information about the secret $\mv$ from the $z$ shares it has access to~\cite{SRK_globecom11, RKSV12}.

\item \textbf{Communication efficiency:~}Note that the secret $\mv$ can be obtained by repairing $t$ code symbols $\cv_{z+1}, \cv_{z+2},\ldots, \cv_{z+t}$ in the associated codeword $\cv \in \Cc$. Since $\Cc$ is an MBMR code, this repair process can be performed by contacting a set of $d$ shares say $\tilde{\cv}_{i_1},\tilde{\cv}_{i_2},\ldots, \tilde{\cv}_{i_d}$ and downloading total 
$$\gamma_t = \frac{\Mc2dt}{(z + t)(2d  - z)}  = \frac{d}{d - z}{\Mc^{s}}$$ symbols (over $\ff_q$) from the $d$ shares we contact (cf.~\ref{eq:wang_bw}). Comparing $\gamma_t$ to the lower bound on \eqref{eq:bw_bound1} establishes the communication efficiency of the secret sharing scheme based on the MBMR code from \cite{WangZhang13}.
\end{itemize}
 
\begin{remark}
Since the secret sharing scheme $\tilde{\Cc}$ is obtained by puncturing $t$ code symbols in the MBMR code $\Cc$. Assuming that the original MBMR code has $d \leq N - t = n - 2t$, we can repair any $t$ shares in a bandwidth efficient manner by invoking the repair mechanism of the MBMR code $\Cc$. This repair process would involve contacting any set of $d$ out of remaining $N - t$ shares and downloading $\gamma_t = \frac{\Mc2dt}{(z + t)(2d  - z)}$ symbols (over $\ff_q$) from the contacted shares.
\end{remark}

\section*{Acknowledgememt}
The authors would like to thank Salim El Rouayheb for pointing us to
the work by Cadambe et al.~\cite{CadambeT13}, which includes a bound and an
existential result on the repair of multiple failures in an MDS code.
\bibliographystyle{unsrt}
\bibliography{multinode}

\appendices

\section{Zigzag codes: simultaneous repair of up to $3$ failed systematic nodes}
\label{appen:zigzag}

\subsection{Description of the zigzag construction~\cite{zigzag13}}
Let $\ZZ_r$ denote the set $\{0, 1,\ldots, r-1\}$. Let $\ev_1, \ev_2,\ldots, \ev_m \in \ZZ_r^{m}$ denote the $m$ standard $m$-dimensional unit vectors. For $i \in [m]$, the vector $\ev_i$ has all but one of its coordinates as zero. The vector $\ev_i$ has value $1$ at its $i$-th coordinate. We use $\ev_0 \in \ZZ_r^m$ to denote an $m$-dimensional all zero vector. For an integer in $[0, r^m - 1]$ we associate a unique vector from $\ZZ_r^m$ as its vector representation. Given $(m+1)r^m$ information symbols, encoding of an $(n = m + 1 + r, k = m + 1)$ zigzag code works as follows.
\begin{itemize}
\item Arrange the $(m+1)r^m$ information symbols in an $r^m \times (m + 1)$ array. For $j \in [m]$ and $i \in [r^m-1]$, let $x_{i,j}$ denote the $(i+1)$-th information symbol in the $(j + 1)$-th column of the array. The $k = (m + 1)$ columns of the information array represent the $k$ systematic nodes in the zigzag code construction.
\item In order to generate $r$ parity nodes, for every $l \in [0, r-1]$ and $s \in [0, r^m-1]$, we define the zigzag set 
\begin{align}
\label{eq:Zset}
\Zc^{l}_{s} = \{x_{i,j}~:i + l\ev_j = s\}.
\end{align}
Note that we use the vector representations  of $i, s \in [0, r^m-1]$ while defining the set $\Zc^{l}_s$ in \eqref{eq:Zset}. Given the zigzag set $\Zc^{l}_s$, $(s+1)$-th symbol stored on the $(l + 1)$-th parity node is linear combination of the information symbols in the zigzag set $\Zc^{l}_s$. The coefficients of the linear combinations belong to non-zero elements (multiplicative group) of a large enough finite field.
\end{itemize}

\subsubsection{Repair of a single systematic node in zigzag codes~\cite{zigzag13}}
Here, we briefly describe the repair mechanism of a single systematic node in the zigzag code construction. For $j \in [1, m]$ and $l \in [0, r-1]$, we define the set 
\begin{align}
\label{eq:Xset}
\Xc^l_{j} = \{i \in [0, r^m - 1]~:~i\cdot \ev_j = r - l\}.
\end{align}
Again, we use the vector representation of the integer $j\in [0, r^m - 1]$ while defining the set $\Xc^{l}_{j}$ in \eqref{eq:Xset}. For $j = 0$ and $l \in [0, r-1]$, we define the corresponding set $\Xc^l_{0}$ as follows. 
\begin{align}
\label{eq:Xset0}
\Xc^l_{0} = \{i \in [0, r^m - 1]~:~i\cdot (1, 1,\ldots, 1) = l\}.
\end{align}
For $j \in [0, m]$, those informations symbols stored on the $(j + 1)$-th systematic node which are indexed by the set $\Xc^l_{j}$ are recovered by downloading the code symbols from the $(l+1)$-th partiy node. From the remaining $k - 1 = m$ systematic nodes, we download those symbols which appear in the symbols downloaded from the parity nodes. Combining the definitions in \eqref{eq:Zset}, \eqref{eq:Xset} and \eqref{eq:Xset0}, we obtain the following. 
\begin{proposition}
\label{prop:zigzag_repair}
For $l \in [0,r-1]$, let $\Dc^{1}_l$ be the set defined as follows. 
\begin{align}
\label{eq:R_repair}
\Dc^{1}_l = \{i \in [0, r^{k-1}-1]~:~i \cdot (1, 1, \ldots, 1) = l\}.
\end{align}
Furthermore, for $j \in [1, k-1]$ and $l \in [0, r-1]$, we define the set $\Dc^{j+1}_{l}$ as follows. 
\begin{align}
\label{eq:R_repair2}
\Dc^{j+1}_l = \{i \in [0, r^{k-1}-1]~:~i \cdot \ev_j = 0~({\rm mod}~r)\}.
\end{align}
Then, for $j \in [0, k-1]$ and $l \in [0, r-1]$, the set $\Dc^{j+1}_{l}$ denotes the indices of the parity symbols downloaded from $(l + 1)$-th parity node in order to repair the $(j+1)$-th systematic node in the event of a single failure.
\end{proposition}

\subsubsection{Structure of symbols downloaded to repair different node in the even of a single node failure}
\label{sec:zigzag_multiple}

In our approach to repair $t$ simultaneous node failures, we contact the remaining $d = n - t$ nodes and download symbols in two stages. In the first stage, for each of the failed $t$ node, we download those $\frac{\alpha}{n - k} = r^{k-2}$ symbols from the contacted node which would have been downloaded to repair this node in the event of single node failure. Since some of the symbols from a helper node contribute to the repair of many nodes during the repair of a single node failure, we end up downloading {\em less than} $\frac{t\alpha}{n - k}$ symbols from each of the $d = n - k$ contacted node. Using the structure of the zigzag code, in the second stage, we then download additional symbols from the helper nodes so that each helper node contributes exactly $\frac{t\alpha}{n - k}$ symbols. In order to identify which symbols need to be downloaded in the second stage we need to understand the structure of the parity symbols downloaded in the first stage. Therefore, we first explore this. 

For the ease of exposition, without loss of generality, we assume that the first $t$ systematic nodes are in failure, i.e., the systematic nodes indexed by the set  $[t] := \{0, 1,\ldots, t-1\}$ experience failure. The analysis for other $t$ systematic nodes can be carried out in a similar manner. Recall that for $j \in [0, t-1]$ and $l \in [0, r-1]$, $\Dc^{j+1}_{l}$ denotes the indices of the parity symbols downloaded from $(l + 1)$-th parity node to repair the $(j + 1)$-th systematic node failure. The sets $\big\{\Dc^{j + 1}_{l}\big\}_{j \in [0, t-1],~l \in [0, r-1]}$ are defined in Proposition~\ref{prop:zigzag_repair}. It follows from the definition of these sets, that for any set of $u$ out $t$ failed nodes, say indexed by the set $\{j_1, j_2,\ldots, j_u\} \subseteq [0, t-1]$, and $l \in [0, r-1]$, we have the following. 
\begin{align}
\label{eq:repair_set_intersection}
\bigcap_{j \in \{j_1, j_2,\ldots, j_u\}}\Dc^{j+1}_{l} = r^{k - 1 - u}.
\end{align}
Given this observations, we now define the following families of sets. For $\{j_1, j_2,\ldots, j_u\} \subseteq [0, t-1]$ and $l \in [0, r-1]$, we define the set $\Uc^{\{j_1+1, j_2+1,\ldots, j_u+1\}}_{l}$ to be the indices of the pairty symbols downloaded from the $(l + 1)$-th parity node which participate in the repair of exactly $u \leq t$ systematic nodes indexed by the set $\{j_1,\ldots, j_u\}$ in the event of single node failure. In particular, given $\Sc \subseteq [t]$ and $l \in [0, r-1]$,we have
\begin{align}
\label{eq:setU}
\Uc^{\Sc}_l = \bigcap_{j \in \Sc}\Dc^{j}_{l} \backslash \bigcup_{\Sc \subsetneq \Sc' \subseteq [t-1]}\bigcap_{j \in \Sc'}\Dc^{j}_{l}.
\end{align}
Combining \eqref{eq:setU} with the definitions of the sets $\{\Dc^{j+1}_{l}\}_{j \in [0, k-1], l \in [0, r-1]}$, we obtain that
\begin{enumerate}
\item {\bf Case~1:~$1 \in \Sc$},
\begin{align}
\label{eq:setU1}
\Uc^{\Sc}_l = \{i~:~i\cdot (1,\ldots,1) = l;~i\cdot \ev_{w-1} = 0~\forall w \in \Sc\backslash\{1\};~\text{and}~i\cdot \ev_{v-1} \neq 0~\forall~v \in [t]\backslash \Sc\}\subseteq [0, r^{k-1}-1].
\end{align}
\item {\bf Case~2:~$1 \in [t]\backslash\Sc$},
\begin{align}
\label{eq:setU2}
\Uc^{\Sc}_l = \{i~:~i\cdot \ev_{w-1} = 0~\forall w \in \Sc;~i\cdot (1,\ldots,1) \neq l;~\text{and}~i\cdot \ev_{v-1} \neq 0~\forall~v \in [2,t]\backslash \Sc\}\subseteq [0, r^{k-1}-1].
\end{align}
\end{enumerate}
Moreover, we have that
\begin{align}
\label{eq:sizeU}
|\Uc^{\Sc}_l|  &=  |\bigcap_{j \in \Sc}\Dc^{j}_{l}| - | \bigcup_{\Sc \subsetneq \Sc' \subseteq [t]}\bigcap_{j \in \Sc'}\Dc^{j}_{l}| \nonumber \\
&\overset{(a)}{=} r^{k - 1 - |\Sc|}\left(1 - \frac{1}{r}\right)^{t - |\Sc|},
\end{align}
where $(a)$ follows from \eqref{eq:setU1} and \eqref{eq:setU2}. 
Note that, by construction, for a fixed $l \in [0, r-1]$, the family of sets $\big\{\Uc^{\Sc}_{l}\big\}_{\Sc \subseteq [0,t-1]}$ comprises disjoint sets. In case of $t = 3$, for a fixed value of $l \in [0:r-1]$, this gives us the following sequences of disjoint sets, $\Uc^{\{1\}}_{l}, \Uc^{\{2\}}_{l}, \Uc^{\{3\}}_{l}, \Uc^{\{1, 2\}}_{l}, \Uc^{\{2, 3\}}_{l}, \Uc^{\{1, 3\}}_{l}, \Uc^{\{1, 2, 3\}}_{l}$. Here, for $j \in [0,2]$, $\Uc^{\{j+1\}}_{l}$ denotes the indices of those symbols from the $(l + 1)$-th parity node which participate only in the repair of $(j+1)$-th systematic node in the event of single node failure. The sets $\Uc^{(1,2)}_l$ represent the sets of symbols from $(l+1)$-th parity node that participate only in the repair of $1$st and $2$nd systematic nodes in the event of single node failure. Similarly, the parity symbols from $(l+1)$th parity nodes that enable repair of each of the first $3$ systematic nodes in the event of a single node failure are represented by $\Uc^{(1,2,3)}_l$.

We now characterize the sets of information symbols on the lost (failed) systematic nodes that participate in the sets $\{ \Uc^{\Sc}_l\}_{\Sc \subseteq [t], l \in [0, r-1]}$. For a given $j \in [0, k-1], l \in [0, r-1]$, and $\Sc \subseteq [t]$, let $\Uc^{\Sc}_{j \rightarrow l}$ denotes the indices of the symbols from the $(j + 1)$-th systematic node which participate in the parity symbols in the set $\Uc^{\Sc}_{l}$. Using \eqref{eq:Zset} and \eqref{eq:setU1}, we can explicitly characterize these sets. In particular, we consider $3$ different case. 
\begin{enumerate}
\item {\bf Case~1~(a)~:~$j = 0$~and~$1 \in \Sc$},
\begin{align}
\label{eq:setUsys1a}
\Uc^{\Sc}_{j \rightarrow l} &\overset{(a)}{=} \big\{i \in [0, r-1]^{k-1}~:~i \in \Uc^{\Sc}_{l}\big\} \nonumber \\
&\overset{(b)}{=} \{i~:~i\cdot (1,\ldots,1) = l;~i\cdot \ev_{w-1} = 0~\forall w \in \Sc\backslash\{1\};~\text{and}~i\cdot \ev_{v-1} \neq 0~\forall~v \in [t]\backslash \Sc\}\subseteq [0, r^{k-1}-1].
\end{align}
\item {\bf Case~1~(b)~:~$j = 0$~and~$1 \in [t]\backslash\Sc$},
\begin{align}
\label{eq:setUays1b}
\Uc^{\Sc}_{j \rightarrow l} = \{i~:~i\cdot \ev_{w-1} = 0~\forall w \in \Sc;~i\cdot (1,\ldots,1) \neq l;~\text{and}~i\cdot \ev_{v-1} \neq 0~\forall~v \in [2,t]\backslash \Sc\}\subseteq [0, r^{k-1}-1].
\end{align}
\item {\bf Case~2~(a)~:~$j \neq 0$~and~$\{1, j + 1\} \subseteq \Sc$,}
\begin{align}
\label{eq:setUsys2a}
\Uc^{\Sc}_{j \rightarrow l} &\overset{(a)}{=} \big\{i \in [0, r-1]^{k-1}~:~i + l\ev_j \in \Uc^{\Sc}_{l}\big\} \nonumber \\
&\overset{(b)}{=} \big\{i~:~i\cdot(1,\ldots, 1) = 0;~i\cdot \ev_j + l= 0;~i \cdot \ev_{w-1} = 0~\forall w \in \Sc \backslash\{1, j+1\};~\text{and}\nonumber \\
&~~~~~~~~~~~~~~~~~~~~~~~~~~~~~~~~~~~~~~~i\cdot \ev_{v-1} \neq 0~\forall~v \in [t]\backslash \Sc\big\}\subseteq [0, r-1]^{k-1}.
\end{align}
\item {\bf Case~2~(b)~:~$j \neq 0$, $1 \in \Sc$~and~$j + 1 \in [t]\backslash \Sc$,}
\begin{align}
\label{eq:setUsys2b}
\Uc^{\Sc}_{j \rightarrow l} &{=} \big\{i~:~i\cdot(1,\ldots, 1) = 0;~i \cdot \ev_{w-1} = 0~\forall w \in \Sc \backslash\{1\};~i\cdot \ev_j + l \neq 0~\text{and}~\nonumber \\
&~~~~~~~~~~~~~~~~~~~~~~~~~~~~~~~~~~~~~~~i\cdot \ev_{v-1} \neq 0~\forall~v \in [t]\backslash\{\Sc \cup \{j+1\}\}\big\} \subseteq [0, r-1]^{k-1}.
\end{align}
\item {\bf Case~2~(c):~$j \neq 0$, $1 \in [t]\backslash \Sc$~and~$j + 1 \in \Sc$,}
\begin{align}
\label{eq:setUsys2c}
\Uc^{\Sc}_{j \rightarrow l} &{=} \big\{i~:~i\cdot \ev_j + l = 0;~i \cdot \ev_{w-1} = 0~\forall w \in \Sc \backslash\{j+1\};~i\cdot(1,\ldots, 1) \neq 0;~\text{and}~\nonumber \\
&~~~~~~~~~~~~~~~~~~~~~~~~~~~~~~~~~~~~~~~i\cdot \ev_{v-1} \neq 0~\forall~v \in [2,t]\backslash \Sc\big\}\subseteq [0, r-1]^{k-1}.
\end{align}
\item {\bf Case~2~(d):~$j \neq 0$~and~$\{1, j + 1\} \in [t]\backslash\Sc$,}
\begin{align}
\label{eq:setUsys2d}
\Uc^{\Sc}_{j \rightarrow l} &{=} \big\{i~:~i \cdot \ev_{w-1} = 0~\forall w \in \Sc; i\cdot(1,\ldots, 1) \neq 0;~i\cdot \ev_j + l\neq 0;~~~\text{and}~\nonumber \\
&~~~~~~~~~~~~~~~~~~~~~~~~~~~~~~~~~~~~~~~i\cdot \ev_{v-1} \neq 0~\forall~v \in [2,t]\backslash \{\Sc \cup \{j+1\}\}\big\}\subseteq [0, r-1]^{k-1}.
\end{align}
\end{enumerate}

\subsection{Repairing $t = 2$ failed nodes}
\label{sec:t2case}

\subsubsection{First stage of the download process}

In the first stage, we download the symbols which enable the repair of $1$st and $2$nd systematic nodes in the event of single node failure. In particular, for $l \in [0, r-1]$, we download the parity symbols indexed by the set $$\Dc^{1}_{l} \cup \Dc^{2}_{l} = \Uc^{\{1\}}_{l} \cup \Uc^{\{2\}}_{l} \cup \Uc^{\{1, 2\}}_{l}$$ from the $(l + 1)$-th parity node. From the remaining $k - 2$ systematic nodes, we download those systematic symbols which appear in these parity nodes.

\bgroup
\def\arraystretch{1.2}
\begin{table}[htbp]
\begin{center}
   \small
    \begin{tabular}{|c|c|c|c|c|c|}
    \hline
    Sets& $l = 0$ &$l = 1$& $\cdots$ & $l = r - 1$ \\ \hline
    $\Uc^{\{1\}}_{l}$&${\color{blue} \Uc^{\{1\}}_{0 \rightarrow 0}}, \Uc^{\{1\}}_{1 \rightarrow 0} = \left(\begin{array}{c}
\Uc^{\{1, 2\}}_{1 \rightarrow 1}\\
\Uc^{\{1, 2\}}_{1 \rightarrow 2}\\
\vdots \\
\Uc^{\{1, 2\}}_{1 \rightarrow r - 1}\\
\end{array} \right)$&${\color{blue}\Uc^{\{1\}}_{0 \rightarrow 1}}, \Uc^{\{1\}}_{1 \rightarrow 1} = \left(\begin{array}{c}
\Uc^{\{1, 2\}}_{1 \rightarrow 0}\\
\Uc^{\{1, 2\}}_{1 \rightarrow 2}\\
\vdots \\
\Uc^{\{1, 2\}}_{1 \rightarrow r - 1}\\
\end{array} \right)$&$\cdots$ &${\color{blue}\Uc^{\{1\}}_{0 \rightarrow r-1}}, \Uc^{\{1\}}_{1 \rightarrow r-1} = \left(\begin{array}{c}
\Uc^{\{1, 2\}}_{1 \rightarrow 0}\\
\Uc^{\{1, 2\}}_{1 \rightarrow 1}\\
\vdots \\
\Uc^{\{1, 2\}}_{1 \rightarrow r - 2}\\
\end{array} \right)$ \\ \hline

   $\Uc^{\{2\}}_{l}$&$\Uc^{\{2\}}_{0 \rightarrow 0} = \left(\begin{array}{c}
{\color{blue} \Uc^{\{1, 2\}}_{0 \rightarrow 1}}\\
\Uc^{\{1, 2\}}_{0 \rightarrow 2}\\
\vdots \\
{\Uc^{\{1, 2\}}_{0 \rightarrow r - 1}}\\
\end{array} \right), {\color{red}\Uc^{\{2\}}_{1 \rightarrow 0}}$&$\Uc^{\{2\}}_{0 \rightarrow 1} = \left(\begin{array}{c}
{\Uc^{\{1, 2\}}_{0 \rightarrow 0}}\\
{\color{blue}\Uc^{\{1, 2\}}_{0 \rightarrow 2}}\\
\vdots \\
\Uc^{\{1, 2\}}_{0 \rightarrow r - 1}\\
\end{array} \right), {\color{red}\Uc^{\{2\}}_{1 \rightarrow 1}}$&$\cdots$ &$\Uc^{\{2\}}_{0 \rightarrow r-1} = \left(\begin{array}{c}
{\color{blue}\Uc^{\{1, 2\}}_{0 \rightarrow 0}}\\
\Uc^{\{1, 2\}}_{0 \rightarrow 1}\\
\vdots \\
{\Uc^{\{1, 2\}}_{0 \rightarrow r - 2}}\\
\end{array} \right), {\color{red}\Uc^{\{2\}}_{1 \rightarrow r-1}}$ \\ \hline
    $\Uc^{\{1,2\}}_{l}$&$\Uc^{\{1,2\}}_{0 \rightarrow 0}, {\color{red} \Uc^{\{1,2\}}_{1 \rightarrow 0}}$ & $\Uc^{\{1,2\}}_{0 \rightarrow 1},  {\color{red}\Uc^{\{1,2\}}_{1 \rightarrow 1}}$ & $\cdots$ & $\Uc^{\{1,2\}}_{0 \rightarrow r-1},  {\color{red}\Uc^{\{1,2\}}_{1 \rightarrow r-1}}$ \\ \hline
   \end{tabular}
\end{center}
\caption{Composition of parity symbols downloaded in the first stage of the repair process to repair of $1$st and $2$nd systematic nodes. Blue color symbols correspond to the matched symbols from the $1$st systematic node. Red color symbols represent the matched symbols from the $2$nd system node.}
\label{tab:example1b}
\end{table}
\egroup

We next illustrate a strategy to match the symbols from the failed systematic nodes using the downloaded symbols illustrated in Table~\ref{tab:example1b}. Formally, we have the following. 

\begin{itemize}
\item \textbf{Matching symbols from the $2$nd systematic node using parity symbols $\big\{\Uc^{\{1, 2\}}_{l}\big\}$:}~We use the symbols for the $(l + 1)$-th parity node which are indexed by the set $\Uc^{\{1, 2\}}_{l}$ to match those symbols from the $2$nd systematic node that are indexed by the set 
$$
\Uc^{\{1, 2\}}_{1 \rightarrow l} = \{i~:~i\cdot(1,\ldots,1) = 0, i\cdot\ev_1 = r - l\}.
$$

\item \textbf{Matching symbols from the $1$st systematic node using parity symbols $\big\{\Uc^{\{1\}}_{l}\big\}$:}~
We use the symbols for the $(l + 1)$-th parity node which are indexed by the set $\Uc^{\{1\}}_{l}$ to match those symbols from the $1$st systematic node that are indexed by the set 
$$
\Uc^{\{1\}}_{0 \rightarrow l} = \{i~:~i\cdot(1,\ldots,1) = l, i\cdot\ev_1 \neq 0\}.
$$

\item \textbf{Matching symbols from the $1$st and the $2$nd systematic node using parity symbols $\big\{\Uc^{\{2\}}_{l}\big\}$:}~
We use the symbols for the $(l + 1)$-th parity node which are indexed by the set $\Uc^{\{1\}}_{l}$ to match those symbols from the $1$st systematic node that are indexed by the set 
\begin{align}
\Uc^{\{1, 2\}}_{0 \rightarrow l + 1} = \{i:~i\cdot(1,\ldots, 1) = l + 1, i\cdot \ev_1 = 0\}.
\end{align}
We use the remaining $\mid \Uc^{\{2\}}_{1 \rightarrow l} \mid - \mid \Uc^{\{1, 2\}}_{0 \rightarrow (l + 1)} \mid = (r^{k-2} - 2r^{k-3)})$ parity symbols to match the symbols from the $2$nd systematic node with the following indices.
\begin{align}
\label{eq:t2_matched2_tmp}
\{i:~i\cdot(1,\ldots, 1) \notin \{0, 1\}, i\cdot \ev_1 = r - l\} \subseteq \Uc^{\{2\}}_{1 \rightarrow l} = \{i:~i\cdot(1,\ldots, 1) \neq 0, i\cdot \ev_1 = r - l\} .
\end{align}
\end{itemize}

As it is clear from Table~\ref{tab:example1b} and the matching processing described above, the following number of symbols from the $2$ failed nodes are matched by the parity symbols downloaded in the first stage.
\begin{enumerate}
\item {\bf Number of symbols matched from $1$st systematic node (blue colored):} 

\begin{align}
\underbrace{\sum_{l = 0}^{r - 1}\left |\Uc^{\{1\}}_{0 \rightarrow l} \right|}_{\text{from~$\Uc^{\{1\}}_{l}$}} + \underbrace{\sum_{l = 0}^{r - 1}\left|\Uc^{\{1, 2\}}_{0 \rightarrow l+1}\right|}_{\text{from~$\Uc^{\{2\}}_{l}$}} = r\cdot \left(r^{k-2} - r^{k - 3}\right) + r \cdot r^{k - 3} = r^{k-1} = \alpha.
\end{align}

\item {\bf Number of symbols matched from $2$nd systematic node (red colored):} 
\begin{align}
\label{eq:second_mt}
\underbrace{\sum_{l = 0}^{r - 1}\left (\left |\Uc^{\{2\}}_{1 \rightarrow l}|\right| -  \left|\Uc^{\{1, 2\}}_{0 \rightarrow  l+1}\right| \right)}_{\text{from~$\Uc^{\{2\}}_{l}$}} + \underbrace{\sum_{l = 0}^{r - 1}\left|\Uc^{\{1, 2\}}_{1 \rightarrow l}\right|}_{\text{from~$\Uc^{\{1,2\}}_{l}$}} = r\cdot \left(r^{k-2} -2 r^{k - 3} \right) + r \cdot r^{k - 3} = r^{k-1} - r^{k-2} = \alpha - r^{k - 2}.
\end{align}
\end{enumerate}

\subsubsection{Second Stage of download process}

\begin{itemize}
\item \textbf{Unmatched symbols from the second systematic node~:}~It follows from \eqref{eq:second_mt} that it remains to match $r^{k-2}$ symbols from the $2$nd systematic node. This requires us to download additional symbols from the intact nodes. Towards this, let's consider the unmatched symbols from the second systematic node in the $(l + 1)$-th parity node (cf. Table~\ref{tab:example1b}). These are exactly those symbols among the symbols indexed by the set $\Uc^{\{2\}}_{1 \rightarrow l}$ which form linear combinations with the symbols indexed by the set $\Uc^{\{1, 2\}}_{0 \rightarrow l + 1}$. It can be easily deduced from \eqref{eq:t2_matched2_tmp} that the unmatched symbols from the second systematic node have the following indices.
\begin{align}
\label{eq:sys12_2}
\Rc_{1 \rightarrow l} &= \{i~:~i\cdot(1,\ldots, 1) = 1,~i\cdot\ev_1  = r - l\}
\end{align}
Another way to see this is as follows. Using \eqref{eq:setUsys1a},  we know that 
\begin{align}
\Uc^{\{1, 2\}}_{0 \rightarrow l + 1} = \{i:~i\cdot(1,\ldots, 1) = l + 1, i\cdot \ev_1 = 0\}.
\end{align}
Utilizing the definition of the zigzag set (cf.~\eqref{eq:Zset}), the symbols from the first systematic nodes which are indexed by the set $\Uc^{\{1, 2\}}_{0 \rightarrow l + 1}$ appear in those parity symbols in the $(l + 1)$-th parity node which are indexed by the set 
\begin{align}
\label{eq:parity12}
\{i~:~i \in \Uc^{\{1, 2\}}_{0 \rightarrow l + 1}\} \subseteq [0: r^k-1].
\end{align}
We again utilize the definition of the zigzag sets (cf.~\eqref{eq:Zset}) to identify the symbols from the second systematic node that appear in those parity symbols from the $(l + 1)$-th parity node which are indexed by the set defined in \eqref{eq:parity12}. These are exactly the symbols indexed by the following set.
\begin{align}
\label{eq:sys12_2}
\Rc_{1 \rightarrow l} &= \{i~:~i\cdot(1,\ldots, 1) + l = l + 1,~i\cdot\ev_1 + l = 0\} \nonumber \\
& = \{i~:~i\cdot(1,\ldots, 1) = 1,~i\cdot\ev_1  = r - l\}
\end{align}
\end{itemize}

\begin{itemize}
\item \textbf{Additional symbols downloaded to match remaining symbols from the second systematic node~:}~Now, let's consider an integer $i^{\ast} \in [0, r^m - 1] = [0, r^{k-1}-1]$ such that the following two conditions hold.
\begin{enumerate}
\item $i^{\ast}\cdot(1,1,\ldots, 1) = 1$. Note that we are using vector representation of $i^{\ast}$ in $\ZZ_r^{k-1} = \ZZ_r^m$ in defining this relationship.
\item $i^{\ast}\cdot\ev_1 = r - 1$, i.e., the first coordinate of the vector representation of $i^{\ast}$ takes the nonzero value $r - 1$.
\end{enumerate}

For $j \in  [2, k-1]$, the set of additional symbols downloaded from the $(j+1)$-th systematic node have their row indices belonging to the following set. 
\begin{align}
\label{eq:Sdef}
\Sc^{\{1, 2\}} = \{i^{\ast} + a_1(\ev_3 - \ev_2) + a_2(\ev_4 - \ev_2) + \cdots + a_{k-3}(\ev_{k-1} - \ev_2)~:~(a_1,\ldots, a_{k-3}) \in [r-1]^{k-3} \}.
\end{align}
Note that we have $|\Sc^{\{1, 2\}}|= r^{k-3}$. The reason behind this particular choice for the set $\Sc^{\{1, 2\}}$ will become clear very soon. Now, let's focus on the additional code symbols that need to be downloaded from the parity nodes. For $l \in [0:r - 1]$, we download those parity symbols from the $(l + 1)$-th parity node which involve the information symbols associated with the set $\Sc^{\{1, 2\}}$. Recall that one can use the definitions of the zigzag sets (cf.~\eqref{eq:Zset}) to identify these additional parity symbols that need to be downloaded. In particular, for $l = 0$, the additional symbols downloaded from the $1$st parity node have their row indices belonging to the set 
\begin{align}
\Pc^{\{1, 2\}}_{0} = \Sc^{\{1, 2\}}.
\end{align}

In general, for $l \in [0, r-1]$, the additional parity symbols downloaded from the $(l + 1)$-th parity nodes have their row indices belonging to the following sets
\begin{align}
\label{eq:Rdef}
\Pc^{\{1, 2\}}_l = \Sc^{\{1, 2\}} + l\cdot \ev_2 = \Sc^{\{1, 2\}} + l\cdot \ev_3 = \cdots = \Sc^{\{1, 2\}} + l\cdot \ev_{k-1}.
\end{align}

Now, let's make sure if these additional parity symbols indexed by the sets $\{\Pc^{\{1, 2\}}_{l}\}_{l \in [0:r-1]}$ indeed help us match the unmatched symbols from the second systematic node. Let's first consider the parity symbols indexed by the set $\Pc^{\{1, 2\}}_0$. The symbols from the second systematic node which can be mathced using these parity symbols are the ones indexed by the set $\Pc^{\{1, 2\}}_0 = \Sc^{\{1, 2\}} \subset [0:r-1]^k$ itself, i.e.,
\begin{align}
\{i~:~i\cdot(1, 1,\ldots, 1) = 1,~i\cdot\ev_1 = r - 1\}.
\end{align}
Note that these are exactly those symbol which remained unmatched as they appear together those symbols from the first systematic nodes that are indexed by the set $\Uc^{\{1, 2\}}_{0 \rightarrow 2}$ in the $2$nd parity node, i.e., the symbols from the second systematic node which are indexed by the set $\Rc_{1 \rightarrow 1}$ (cf.~\ref{eq:sys12_2}). Similarly, one can show that the symbols from the second systematic node which can potentially be matched using the additional symbols downloaded from the $(l + 1)$-th parity node, i.e., the parity symbols indexed by the set $\Pc^{\{1, 2\}}_{l}$, are associated with the set. 
\begin{align}
\big\{i~:~i + l\ev_1 \in \Pc^{\{1, 2\}}_l\big\} = \{i~:~i\cdot(1, 1,\ldots, 1) = 1,~i\cdot\ev_1 = r - (l + 1)\}.
\end{align}
Note that these are exactly those symbol denoted by the set $\Rc_{1 \rightarrow l + 1}$. That is, the symbols from second systematic node which remained unmatched as they appear together those symbols from the first systematic nodes that are indexed by the set $\Uc^{\{1, 2\}}_{0 \rightarrow l+2}$ in the $(l + 2)$-th parity node. 
\end{itemize}

\subsection{Repairing $t = 3$ failed nodes}
\label{sec:t3case}

\subsubsection{First stage of the download process}

In the first stage, we download the symbols which enable the repair of the first $3$ systematic nodes in the event of single node failure. In particular, for $l \in [0, r-1]$, we download the parity symbols indexed by the set $$\Dc^{1}_{l} \cup \Dc^{2}_{l} \cup \Dc^{3}_{l} = \Uc^{\{1\}}_{l} \cup \Uc^{\{2\}}_{l} \cup \Uc^{\{3\}}_{l}\cup \Uc^{\{1, 2\}}_{l}\cup \Uc^{\{2, 3\}}_{l}\cup \Uc^{\{1,3\}}_{l}\cup \Uc^{\{1, 2,3\}}_{l}$$ from the $(l + 1)$-th parity node. From the remaining $k - 2$ systematic nodes, we download those systematic symbols which appear in these parity nodes.

\bgroup
\def\arraystretch{1.2}
\begin{table}[htbp]
\begin{center}
   \tiny
    \begin{tabular}{| c | c | c | c | c | c |}
    \hline
    {\footnotesize Sets} & {\footnotesize$l = 0$} &{\footnotesize$l = 1$}& {\footnotesize$\cdots$} & {\footnotesize$l = r - 1$} \\ \hline
$\Uc^{\{1\}}_{l}$&${\color{blue} \Uc^{\{1\}}_{0 \rightarrow 0}}, \left(\begin{array}{c}
\Uc^{\{1, 2\}}_{1 \rightarrow 1}\\
\Uc^{\{1, 2\}}_{1 \rightarrow 2}\\
\Uc^{\{1, 2\}}_{1 \rightarrow 3}\\
\vdots \\
\Uc^{\{1, 2\}}_{1 \rightarrow r - 1}\\
\end{array} \right), \left(\begin{array}{c}
{\color{green}\Uc^{\{1, 3\}}_{2 \rightarrow 1}}\\
\Uc^{\{1, 3\}}_{2 \rightarrow 2}\\
\Uc^{\{1, 3\}}_{2 \rightarrow 3}\\
\vdots \\
{\Uc^{\{1, 3\}}_{2 \rightarrow r - 1}}\\
\end{array} \right)$&${\color{blue} \Uc^{\{1\}}_{0 \rightarrow 1}}, \left(\begin{array}{c}
\Uc^{\{1, 2\}}_{1 \rightarrow 0}\\
\Uc^{\{1, 2\}}_{1 \rightarrow 2}\\
\Uc^{\{1, 2\}}_{1 \rightarrow 3}\\
\vdots \\
\Uc^{\{1, 2\}}_{1 \rightarrow r - 1}\\
\end{array} \right), \left(\begin{array}{c}
{\Uc^{\{1, 3\}}_{2 \rightarrow 0}}\\
{\color{green} \Uc^{\{1, 3\}}_{2 \rightarrow 2}}\\
{\Uc^{\{1, 3\}}_{2 \rightarrow 3}}\\
\vdots \\
\Uc^{\{1, 3\}}_{2 \rightarrow r - 1}\\
\end{array} \right)$&$\cdots$ &${\color{blue} \Uc^{\{1\}}_{0 \rightarrow r-1}}, \left(\begin{array}{c}
\Uc^{\{1, 2\}}_{1 \rightarrow 0}\\
\Uc^{\{1, 2\}}_{1 \rightarrow 1}\\
\Uc^{\{1, 2\}}_{1 \rightarrow 2}\\
\vdots \\
\Uc^{\{1, 2\}}_{1 \rightarrow r - 2}\\
\end{array} \right), \left(\begin{array}{c}
{\color{green} \Uc^{\{1, 3\}}_{2 \rightarrow 0}}\\
\Uc^{\{1, 3\}}_{2 \rightarrow 1}\\
\Uc^{\{1, 3\}}_{2 \rightarrow 2}\\
\vdots \\
{\Uc^{\{1, 3\}}_{2 \rightarrow r - 2}}\\
\end{array} \right)$ \\ \hline


$\Uc^{\{2\}}_{l}$&$\left(\begin{array}{c}
\Uc^{\{1, 2\}}_{0 \rightarrow 1}\\
{\color{blue}\Uc^{\{1, 2\}}_{0 \rightarrow 2}}\\
{\Uc^{\{1, 2\}}_{0 \rightarrow 2}}\\
\vdots \\
{\Uc^{\{1, 2\}}_{0 \rightarrow r - 1}}\\
\end{array} \right), {\color{red} \Uc^{\{2\}}_{1 \rightarrow 0}}, \left(\begin{array}{c}
\Uc^{\{2, 3\}}_{2 \rightarrow 1}\\
{\color{green} \Uc^{\{2, 3\}}_{2 \rightarrow 2}}\\
\Uc^{\{2, 3\}}_{2 \rightarrow 3}\\
\vdots \\
{\Uc^{\{2, 3\}}_{2 \rightarrow r - 1}}\\
\end{array} \right)$&$\left(\begin{array}{c}
{\Uc^{\{1, 2\}}_{0 \rightarrow 0}}\\
\Uc^{\{1, 2\}}_{0 \rightarrow 2}\\
{\color{blue} \Uc^{\{1, 2\}}_{0 \rightarrow 3}}\\
\vdots \\
\Uc^{\{1, 2\}}_{0 \rightarrow r - 1}\\
\end{array} \right), {\color{red} \Uc^{\{2\}}_{1 \rightarrow 1}}, \left(\begin{array}{c}
{\Uc^{\{2, 3\}}_{2 \rightarrow 0}}\\
\Uc^{\{2, 3\}}_{2 \rightarrow 2}\\
{\color{green} \Uc^{\{2, 3\}}_{2 \rightarrow 3}}\\
\vdots \\
\Uc^{\{2, 3\}}_{2 \rightarrow r - 1}\\
\end{array} \right)$&$\cdots$ &$\left(\begin{array}{c}
\Uc^{\{1, 2\}}_{0 \rightarrow 0}\\
{\color{blue}\Uc^{\{1, 2\}}_{0 \rightarrow 1}}\\
{\Uc^{\{1, 2\}}_{0 \rightarrow 2}}\\
\vdots \\
{\Uc^{\{1, 2\}}_{0 \rightarrow r - 2}}\\
\end{array} \right), {\color{red} \Uc^{\{2\}}_{1 \rightarrow r-1}}, \left(\begin{array}{c}
\Uc^{\{2, 3\}}_{2 \rightarrow 0}\\
{\color{green} \Uc^{\{2, 3\}}_{2 \rightarrow 1}}\\
\Uc^{\{2, 3\}}_{2 \rightarrow 2}\\
\vdots \\
{\Uc^{\{2, 3\}}_{2 \rightarrow r - 2}}\\
\end{array} \right)$ \\ \hline


$\Uc^{\{3\}}_{l}$&$\left(\begin{array}{c}
\Uc^{\{1, 3\}}_{0 \rightarrow 1}\\
\Uc^{\{1, 3\}}_{0 \rightarrow 2}\\
{\color{blue}\Uc^{\{1, 3\}}_{0 \rightarrow 3}}\\
{\Uc^{\{1, 3\}}_{0 \rightarrow 4}}\\
\vdots \\
{\Uc^{\{1, 3\}}_{0 \rightarrow r - 1}}\\
\end{array} \right), \left(\begin{array}{c}
\Uc^{\{2, 3\}}_{1 \rightarrow 1}\\
{\color{red}\Uc^{\{2, 3\}}_{1 \rightarrow 2}}\\
{\Uc^{\{2, 3\}}_{1 \rightarrow 3}}\\
{\Uc^{\{2, 3\}}_{1 \rightarrow 4}}\\
\vdots \\
{\Uc^{\{2, 3\}}_{1 \rightarrow r - 1}}\\
\end{array} \right), {\color{green} \Uc^{\{3\}}_{2 \rightarrow 0}}$&$\left(\begin{array}{c}
\Uc^{\{1, 3\}}_{0 \rightarrow 0}\\
\Uc^{\{1, 3\}}_{0 \rightarrow 2}\\
\Uc^{\{1, 3\}}_{0 \rightarrow 3}\\
{\color{blue}\Uc^{\{1, 3\}}_{0 \rightarrow 4}}\\
\vdots \\
{\Uc^{\{1, 3\}}_{0 \rightarrow r - 1}}\\
\end{array} \right), \left(\begin{array}{c}
\Uc^{\{2, 3\}}_{1 \rightarrow 0}\\
\Uc^{\{2, 3\}}_{1 \rightarrow 2}\\
{\color{red} \Uc^{\{2, 3\}}_{1 \rightarrow 3}}\\
{\Uc^{\{2, 3\}}_{1 \rightarrow 4}}\\
\vdots \\
{\Uc^{\{2, 3\}}_{1 \rightarrow r - 1}}\\
\end{array} \right), {\color{green}\Uc^{\{2\}}_{2 \rightarrow 1}}$&$\cdots$ &$\left(\begin{array}{c}
\Uc^{\{1, 3\}}_{0 \rightarrow 0}\\
\Uc^{\{1, 3\}}_{0 \rightarrow 1}\\
{\color{blue} \Uc^{\{1, 3\}}_{0 \rightarrow 2}}\\
{\Uc^{\{1, 3\}}_{0 \rightarrow 3}}\\
\vdots \\
{ \Uc^{\{1, 3\}}_{0 \rightarrow r - 2}}\\
\end{array} \right),\left(\begin{array}{c}
\Uc^{\{2, 3\}}_{1 \rightarrow 0}\\
{\color{red} \Uc^{\{2, 3\}}_{1 \rightarrow 1}}\\
{\Uc^{\{2, 3\}}_{1 \rightarrow 2}}\\
{\Uc^{\{2, 3\}}_{1 \rightarrow 3}}\\
\vdots \\
{\Uc^{\{2, 3\}}_{1 \rightarrow r - 2}}\\
\end{array} \right), {\color{green} \Uc^{\{2\}}_{2 \rightarrow r-1}}$ \\ \hline

$\Uc^{\{1,2\}}_{l}$&$\Uc^{\{1, 2\}}_{0 \rightarrow 0}, {\color{red}\Uc^{\{1, 2\}}_{1 \rightarrow 0}}, \left(\begin{array}{c}
\Uc^{\{1, 2, 3\}}_{2 \rightarrow 1}\\
\Uc^{\{1, 2, 3\}}_{2 \rightarrow 2}\\
\vdots \\
\Uc^{\{1, 2, 3\}}_{2 \rightarrow r - 1}\\
\end{array} \right)$&$\Uc^{\{1, 2\}}_{0 \rightarrow 1}, {\color{red} \Uc^{\{1, 2\}}_{1 \rightarrow 1}}, \left(\begin{array}{c}
\Uc^{\{1, 2, 3\}}_{2 \rightarrow 0}\\
\Uc^{\{1, 2, 3\}}_{2 \rightarrow 2}\\
\vdots \\
\Uc^{\{1, 2, 3\}}_{2 \rightarrow r - 1}\\
\end{array} \right)$&$\cdots$ &$\Uc^{\{1, 2\}}_{0 \rightarrow r-1}, {\color{red} \Uc^{\{1, 2\}}_{1 \rightarrow r-1}}, \left(\begin{array}{c}
\Uc^{\{1, 2, 3\}}_{2 \rightarrow 0}\\
\Uc^{\{1, 2, 3\}}_{2 \rightarrow 1}\\
\vdots \\
\Uc^{\{1, 2, 3\}}_{2 \rightarrow r - 2}\\
\end{array} \right)$ \\ \hline


$\Uc^{\{2, 3\}}_{l}$&$\left(\begin{array}{c}
{\color{blue} \Uc^{\{1, 2, 3\}}_{0 \rightarrow 1}}\\
\Uc^{\{1, 2, 3\}}_{0 \rightarrow 2}\\
\vdots \\
{ \Uc^{\{1, 2, 3\}}_{0 \rightarrow r - 1}}\\
\end{array} \right), \Uc^{\{2, 3\}}_{1 \rightarrow 0}, {\color{green} \Uc^{\{2, 3\}}_{2 \rightarrow 0}}$&$\left(\begin{array}{c}
{\Uc^{\{1, 2, 3\}}_{0 \rightarrow 0}}\\
{\color{blue} \Uc^{\{1, 2, 3\}}_{0 \rightarrow 2}}\\
\vdots \\
\Uc^{\{1, 2, 3\}}_{0 \rightarrow r - 1}\\
\end{array} \right), \Uc^{\{2, 3\}}_{1 \rightarrow 1}, {\color{green} \Uc^{\{2, 3\}}_{2 \rightarrow 1}}$&$\cdots$ &$\left(\begin{array}{c}
{\color{blue}\Uc^{\{1, 2, 3\}}_{0 \rightarrow 0}}\\
\Uc^{\{1, 2, 3\}}_{0 \rightarrow 1}\\
\vdots \\
{\Uc^{\{1, 2, 3\}}_{0 \rightarrow r - 2}}\\
\end{array} \right), \Uc^{\{2, 3\}}_{1 \rightarrow r-1}, {\color{green} \Uc^{\{2, 3\}}_{2 \rightarrow r-1}}$ \\ \hline


${\Uc^{\{1, 3\}}_{l}}$&${\color{blue}\Uc^{\{1, 3\}}_{0 \rightarrow 0}}, \left(\begin{array}{c}
\Uc^{\{1, 2, 3\}}_{1 \rightarrow 1}\\
\Uc^{\{1, 2, 3\}}_{1 \rightarrow 2}\\
{\color{red}\Uc^{\{1, 2, 3\}}_{1 \rightarrow 3}}\\
\Uc^{\{1, 2, 3\}}_{1 \rightarrow 4}\\
\vdots \\
{ \Uc^{\{1, 2, 3\}}_{1 \rightarrow r - 1}}\\
\end{array} \right), { \Uc^{\{1,3\}}_{2 \rightarrow 0}}$&${\color{blue} \Uc^{\{1, 3\}}_{0 \rightarrow 1}}, \left(\begin{array}{c}
{\Uc^{\{1, 2, 3\}}_{1 \rightarrow 0}}\\
\Uc^{\{1, 2, 3\}}_{1 \rightarrow 2}\\
\Uc^{\{1, 2, 3\}}_{1 \rightarrow 3}\\
{\color{red} \Uc^{\{1, 2, 3\}}_{1 \rightarrow 4}}\\
\vdots \\
{\Uc^{\{1, 2, 3\}}_{1 \rightarrow r - 1}}\\
\end{array} \right), {\Uc^{\{1, 3\}}_{2 \rightarrow 1}}$&$\cdots$ &${\color{blue} \Uc^{\{1, 3\}}_{0 \rightarrow r-1}},\left(\begin{array}{c}
\Uc^{\{1, 2, 3\}}_{1 \rightarrow 0}\\
\Uc^{\{1, 2, 3\}}_{1 \rightarrow 1}\\
{\color{red}\Uc^{\{1, 2, 3\}}_{1 \rightarrow 2}}\\
\Uc^{\{1, 2, 3\}}_{1 \rightarrow 3}\\
\vdots \\
{\Uc^{\{1, 2, 3\}}_{1 \rightarrow r - 2}}\\
\end{array} \right), {\Uc^{\{1, 3\}}_{2 \rightarrow r-1}}$ \\ \hline

    $\Uc^{\{1,2, 3\}}_{l}$&$\Uc^{\{1,2, 3\}}_{0 \rightarrow 0}, \Uc^{\{1,2, 3\}}_{1 \rightarrow 0}, {\color{green} \Uc^{\{1,2, 3\}}_{2 \rightarrow 0}}$ & $\Uc^{\{1,2, 3\}}_{0 \rightarrow 1}, \Uc^{\{1,2, 3\}}_{1 \rightarrow 1}, {\color{green} \Uc^{\{1,2, 3\}}_{2 \rightarrow 1}}$ & $\cdots$ & $\Uc^{\{1,2, 3\}}_{0 \rightarrow r-1}, \Uc^{\{1,2, 3\}}_{1 \rightarrow r-1}, {\color{green} \Uc^{\{1,2, 3\}}_{2 \rightarrow r-1}}$ \\ \hline
   \end{tabular}
\end{center}
\caption{Composition of parity symbols downloaded during repair of $t=3$ node failures.  Blue colored symbols correspond to the matched symbols from the $1$st systematic node. Red colored symbols represent the matched symbols from the $2$nd system node. Green colored symbols are used to denote the matched symbols from the $3$rd systematic node.}
\label{tab:example2b}
\end{table}
\egroup

\begin{itemize}
\item \textbf{Matching symbols from first and second systematic node using parity symbols $\big\{\Uc^{\{1, 2, 3\}}_{l}\big\}$:}~We use the symbols for the $(l + 1)$-th parity node which are indexed by the set $\Uc^{\{1, 2, 3\}}_{l}$ to match those symbols from the third systematic node that are indexed by the set 
$$
\Uc^{\{1, 2, 3\}}_{2 \rightarrow l} = \{i~:~i\cdot(1,\ldots,1) = 0, i\cdot\ev_1 = 0, i\cdot \ev_2 = r - l\}.
$$
\end{itemize}

\begin{itemize}
\item \textbf{Matching symbols from first and second systematic node using parity symbols $\big\{\Uc^{\{1, 3\}}_{l}\big\}$:}~We propose the following matching scheme for the symbols from the first and the second systematic node. Given the parity symbols from the $(l+1)$-th parity node which are indexed by the set $\Uc^{\{1, 3\}}_l$, we use them to match those symbols from the second systematic node which are indexed by the set 
\begin{align}
\Uc^{\{1, 2, 3\}}_{1 \rightarrow l + 3} = \{i~:~i\cdot(1,\ldots, 1) = 0, i\cdot\ev_1 = r - (l + 3), i\cdot \ev_2 = 0\}.
\end{align}
This would allow us to use the remaining $\mid\Uc^{\{1, 3\}}_l\mid - \mid\Uc^{\{1, 2, 3\}}_{1 \rightarrow l + 3}\mid = (r - 2)r^{k-4}$ parity symbols indexed by the set $\Uc^{\{1, 3\}}_l$ to match those symbols from the second systematic node which are indexed by the following set.
\begin{align}
\{i~:~i\cdot(1,\ldots, 1) = l, i\cdot\ev_1 \notin \{0, r-3\}, i\cdot\ev_2 = 0\} \subset \Uc^{\{1, 3\}}_{0 \rightarrow l} = \{i~:~i\cdot(1,\ldots, 1) = l, i\cdot\ev_1 \neq 0, i\cdot\ev_2 = 0\}.
\end{align}

\item \textbf{Matching symbols from first and third systematic node using parity symbols $\big\{\Uc^{\{2, 3\}}_{l}\big\}$:}~We propose the following matching scheme for the symbols from the first and the second systematic node. Given the parity symbols from the $(l+1)$-th parity node which are indexed by the set $\Uc^{\{2, 3\}}_l$, we use them to match those symbols from the first systematic node which are indexed by the set 
\begin{align}
\Uc^{\{1, 2, 3\}}_{0 \rightarrow l + 1} = \{i~:~i\cdot(1,\ldots, 1) = l + 1, i\cdot\ev_1 = 0, i\cdot \ev_2 = 0\}.
\end{align}
This would allow us to use the remaining $\mid\Uc^{\{2, 3\}}_l\mid - \mid\Uc^{\{1, 2, 3\}}_{0 \rightarrow l +1}\mid = (r - 2)r^{k-4}$ parity symbols indexed by the set $\Uc^{\{2, 3\}}_l$ to match those symbols from the third systematic node which are indexed by the following set.
\begin{align}
\{i~:~i\cdot(1,\ldots, 1) \notin \{0, 1\}, i\cdot\ev_1 = 0, i\cdot\ev_2 = r - 2\} \subset \Uc^{\{2, 3\}}_{2 \rightarrow l} = \{i~:~i\cdot(1,\ldots, 1) \neq 0, i\cdot\ev_1  = 0, i\cdot\ev_2 = 0\}.
\end{align}

\item \textbf{Matching symbols from second systematic node using parity symbols $\big\{\Uc^{\{1, 2\}}_{l}\big\}$:}~Given the parity symbols from the $(l+1)$-th parity node which are indexed by the set $\Uc^{\{1, 2\}}_l$, we use them to match those symbols from the second systematic node which are indexed by the set 
\begin{align}
\Uc^{\{1, 2\}}_{1 \rightarrow l} = \{i~:~i\cdot(1,\ldots, 1) = 0, i\cdot\ev_1 = r - l, i\cdot \ev_2 = 0\}.
\end{align}
\end{itemize}

\begin{itemize}
\item \textbf{Matching symbols from first and third systematic node using parity symbols $\big\{\Uc^{\{1\}}_{l}\big\}$:}~
We utilize the parity symbols from the $(l + 1)$-th parity node which are indexed by the set $\{Uc^{\{1\}}_{l}$ to match the symbols from the third systematic node with the following indices.
\begin{align}
\Uc^{\{1, 3\}}_{2 \rightarrow l+1} = \{i~:~i\cdot(1,\ldots, 1) = 0, i\cdot \ev_1 \neq 0, i \cdot \ev_2 = r - (l + 1)\}.
\end{align}
The remaining $\mid\Uc^{\{1\}}_{l}\mid - \mid \Uc^{\{1, 3\}}_{2 \rightarrow l+1} \mid = r^{k-4}(r - 1)^2 - r^{k-4}(r-1)$ unused parity symbols indexed by the set $\Uc^{\{1\}}_{l}$ are used to match the symbols from the first systematic appearing in those parity symbols.

\item \textbf{Matching symbols from first and third systematic node using parity symbols $\big\{\Uc^{\{2\}}_{l}\big\}$:}
Using the parity symbols from the $(l+1)$-th parity node which are indexed by the set $\Uc^{\{2\}}_{l}$, we match the symbols from the first systematic node with the following indices.
\begin{align}
\Uc^{\{1, 2\}}_{0 \rightarrow l+2} = \{i~:~i\cdot(1,\ldots, 1) = l + 2, i\cdot \ev_1 = 0, i \cdot \ev_2 \neq 0\}. \\
\end{align}
In addition, we also use these parity symbols to match the symbols from the third systematic node with the following indices.
\begin{align}
\widehat{\Uc}^{\{2, 3\}}_{2 \rightarrow l + 2} = \{i~:~i\cdot(1,\ldots,1) = 1, i\cdot\ev_1 = 0, i\cdot\ev_2 = r - (l+2)\} \subset \Uc^{\{2, 3\}}_{2 \rightarrow l + 2}.
\end{align}
This leaves us with $\mid\Uc^{\{2\}}_{l}\mid - \mid\Uc^{\{1, 2\}}_{0 \rightarrow l+2}\mid - \mid\widehat{\Uc}^{\{2, 3\}}_{2 \rightarrow l + 2}\mid =  r^{k-4}(r - 1)^2 - r^{k-4}(r - 1) - r^{k-4}$ unused parity symbols, which we use to match the symbols from the second systematic node appearing in those parity symbols. 

\item \textbf{Matching symbols from first and third systematic node using parity symbols $\big\{\Uc^{\{3\}}_{l}\big\}$:}
Using the parity symbols from the $(l+1)$-th parity node which are indexed by the set $\Uc^{\{3\}}_{l}$, we match the symbols from the second systematic node with the following indices.

\begin{align}
\Uc^{\{2, 3\}}_{1 \rightarrow l+2} = \{i~:~i\cdot(1,\ldots, 1) \neq 0, i\cdot \ev_1 = r - (l + 2), i \cdot \ev_2 = 0\}.
\end{align}
In addition, we also use these parity symbols to match the symbols from the first systematic node with the following indices.
\begin{align}
\widehat{\Uc}^{\{1, 3\}}_{0 \rightarrow l+3}  \{i~:~i\cdot(1,\ldots, 1) = l + 3, i\cdot \ev_1 = r - 3, i \cdot \ev_2 = 0\} \subset \Uc^{\{1, 3\}}_{0 \rightarrow l+3}.
\end{align}
We utilize the remaining $\mid\Uc^{\{3\}}_{l}\mid - \mid\Uc^{\{2, 3\}}_{1 \rightarrow l+2}\mid - \mid\widehat{\Uc}^{\{1, 3\}}_{0 \rightarrow l+3}\mid =  r^{k-4}(r - 1)^2 - r^{k-4}(r - 1) - r^{k-4}$ unused parity symbols, which we use to match the symbols from the second systematic node appearing in those parity symbols. 
\end{itemize}

This concludes the first stage of the downloading process and we have utilized all the parity symbols downloaded in the first stage to match certain systematic symbols corresponding to the three failed nodes. We now move to the second stage of the download process where we download additional symbols in order to match the remaining unmatched symbols associated with the three failed systematic nodes. We illustrate our strategy to match the symbols from the failed systematic nodes using the downloaded symbols during the first stage in Table~\ref{tab:example2b}. Let's count the number of symbols from different failed systematic nodes that are matched according to Table~\ref{tab:example2b}.
\begin{enumerate}
\item {\bf Symbols from $1$st systematic node (blue colored):} 
\begin{align}
&\underbrace{\sum\limits_{l = 0}^{r - 1}\left(\left|\Uc^{\{1,3\}}_{0 \rightarrow l}\right| - \left|\Uc^{\{1,2,3\}}_{1 \rightarrow l+3}\right|\right)}_{\text{from~$\Uc^{\{1, 3\}}_{l}$}} + \underbrace{\sum\limits_{l = 0}^{r - 1}\left|\Uc^{\{1,2,3\}}_{0 \rightarrow l+1}\right|}_{\text{from~$\Uc^{\{2, 3\}}_{l}$}}+  \underbrace{\sum\limits_{l = 0}^{r-1}\left|\widehat{\Uc}^{\{1, 3\}}_{0 \rightarrow l + 3}\right|}_{\text{from~$\Uc^{\{3\}}_{l}$}} + \underbrace{\sum\limits_{l = 0}^{r-1}\left|\Uc^{\{1,2\}}_{0 \rightarrow l+2}\right|}_{\text{from~$\Uc^{\{2\}}_{l}$}} + \underbrace{\sum\limits_{l = 0}^{r-1}\left( \left|\Uc^{\{1\}}_{0 \rightarrow l}\right| - \left|\Uc^{\{1,3\}}_{2 \rightarrow l+1}\right| \right)}_{\text{from~$\Uc^{\{1\}}_{l}$}} \nonumber \\
&= \left(r^{k - 2} - r^{k - 3} - r^{k-3}\right) + r^{k - 3} + r^{k - 3} + \left(r^{k - 2} - r^{k - 3}\right) + \left(r^{k - 1} - 2r^{k - 2} + r^{k- 3}  - (r^{k - 2} - r^{k-3})\right) \nonumber \\
&= r^{k-1} - (r^{k - 2} - r^{k-3}).
\end{align}

\item {\bf Symbols from $2$nd systematic node (red colored):}
\begin{align}
&\underbrace{\sum\limits_{l = 0}^{r - 1}\left|\Uc^{\{1,2,3\}}_{1 \rightarrow l + 3}\right|}_{\text{from~$\Uc^{\{1, 3\}}_{l}$}}+  \underbrace{\sum\limits_{l = 0}^{r-1}\left|\Uc^{\{1, 2\}}_{1 \rightarrow l}\right|}_{\text{from~$\Uc^{\{1,2\}}_{l}$}} + \underbrace{\sum\limits_{l = 0}^{r-1}\left|\Uc^{\{2,3\}}_{1 \rightarrow l + 2}\right|}_{\text{from~$\Uc^{\{3\}}_{l}$}} + \underbrace{\sum\limits_{l = 0}^{r-1}\left( \left|\Uc^{\{2\}}_{1 \rightarrow l}\right| - \left|\Uc^{\{1,2\}}_{0 \rightarrow l+2}\right| -  \left|\widehat{\Uc}^{\{2,3\}}_{1 \rightarrow l+2}\right| \right)}_{\text{from~$\Uc^{\{2\}}_{l}$}} \nonumber \\
&= r^{k - 3} + \left(r^{k - 2} - r^{k - 3}\right) + \left(r^{k - 2} - r^{k - 3})\right) + \left((r^{k - 1} - 2r^{k - 2} + r^{k- 3} )- (r^{k - 2} -  r^{k - 3}) - r^{k - 3}\right) \nonumber \\
&= r^{k-1} - r^{k - 2}.
\end{align}

\item {\bf Symbols from $3$rd systematic node (green colored):}
\begin{align}
&\underbrace{\sum\limits_{l = 0}^{r-1}\left|\Uc^{\{1, 2, 3\}}_{2 \rightarrow l}\right|}_{\text{from~$\Uc^{\{1,2,3\}}_{l}$}}+ \underbrace{\sum\limits_{l = 0}^{r - 1}\left(\left|\Uc^{\{2, 3\}}_{2 \rightarrow l}\right| - \left|\Uc^{\{1,2,3\}}_{0 \rightarrow l+1}\right|\right)}_{\text{from~$\Uc^{\{2, 3\}}_{l}$}}+ \nonumber \\
& \underbrace{\sum\limits_{l = 0}^{r-1}\left(\left|\Uc^{\{3\}}_{2 \rightarrow l}\right| - \left|\widehat{\Uc}^{\{1,3\}}_{0 \rightarrow l+3}\right| - \left|\Uc^{\{2, 3\}}_{1 \rightarrow l+2}\right|\right)}_{\text{from~$\Uc^{\{3\}}_{l}$}} +  \underbrace{\sum\limits_{l = 0}^{r-1}\left|\Uc^{\{1,3\}}_{2 \rightarrow l+1}\right|}_{\text{from~$\Uc^{\{1\}}_{l}$}} + \underbrace{\sum\limits_{l = 0}^{r-1} \left|\widehat{\Uc}^{\{2,3\}}_{1 \rightarrow l+2}\right|}_{\text{from~$\Uc^{\{2\}}_{l}$}} \nonumber \\
&= r^{k - 3} + \left((r^{k-2} - r^{k - 3})  - r^{k-3}\right)  + \nonumber \\ 
&~~~\left((r^{k-1} - 2r^{k - 2} + r^{k - 3}) - r^{k - 3}  - (r^{k - 2} - r^{k-3})  \right)  +\left(r^{k-2} -  r^{k - 3}\right) + r^{k - 3} \nonumber \\
&=  r^{k-1} - r^{k - 2}.
\end{align}
\end{enumerate} 

\subsubsection{Second Stage of download process}

First, let's identify the unmatched systematic symbols at the end of the first stage of the downloading process.

\begin{itemize}

\item \textbf{Unmatched symbols from the first systematic node:}~The symbols from the second systematic node that are matched using the parity symbols from the $(l + 1)$-th parity node are indexed by the set
\begin{align}
\Uc^{\{1, 3\}}_{2 \rightarrow l+1} = \{i~:~i\cdot(1,\ldots, 1) = 0, i\cdot \ev_1 \neq 0, i \cdot \ev_2 = r - (l + 1)\}.
\end{align}
Using \eqref{eq:Zset}, we can identify the indices of the partiy symbols from the $(l + 1)$-parity node where these symbols participate is as follows.
\begin{align}
\Zc^{\{1, 3\}}_{l, (2 \rightarrow l +1)} = \Uc^{\{1, 3\}}_{2 \rightarrow l+1} + l\ev_2 = \{i~:~i\cdot(1,\ldots, 1) = l, i\cdot \ev_1 \neq 0, i \cdot \ev_2 = r -  1\}.
\end{align}
The symbols from the first systematic node which remain unmatched at the end of first stage (and require downloading additional symbols) due to their participation in the parity symbols indexed by the set $\Zc^{\{1, 3\}}_{l, (2 \rightarrow l +1)} $ in the $(l + 1)$-th parity node are as follows. 
\begin{align}
\label{eq:R31}
\Rc_{0 \rightarrow l} = \{i~:~i \in \Zc^{\{1, 3\}}_{l, (2 \rightarrow l +1)}\} =  \{i~:~i\cdot(1,\ldots, 1) = l, i\cdot \ev_1 \neq 0, i \cdot \ev_2 = r -  1\}.
\end{align}

\item \textbf{Unmatched symbols from the second systematic node:}~The symbols from the first systematic node that are (potentially) matched using the parity symbols downloaded from the $(l + 1)$-th parity node during the first stage are indexed by the following two sets. 
\begin{align}
\Uc^{\{1, 2\}}_{0 \rightarrow l+2} = \{i~:~i\cdot(1,\ldots, 1) = l + 2, i\cdot \ev_1 = 0, i \cdot \ev_2 \neq 0\}. \\
\widetilde{\Uc}^{\{2, 3\}}_{2 \rightarrow l + 2} = \{i~:~i\cdot(1,\ldots,1) = 1, i\cdot\ev_1 = 0, i\cdot\ev_2 = r - (l+2)\} \subset {\Uc}^{\{2, 3\}}_{2 \rightarrow l + 2}.
\end{align}
Using \eqref{eq:Zset}, we can identify the indices of the partiy symbols from the $(l + 1)$-parity node where these symbols participate as follows.
\begin{align}
\Zc^{\{1, 2\}}_{l, (0 \rightarrow l +2)} = \Uc^{\{1, 2\}}_{0 \rightarrow l+2} = \{i~:~i\cdot(1,\ldots, 1) = l + 2, i\cdot \ev_1 = 0, i \cdot \ev_2 \neq 0\}. \\
\widetilde{\Zc}^{\{2, 3\}}_{l, (2 \rightarrow l + 2)} = \widetilde{\Uc}^{\{2, 3\}}_{2 \rightarrow l + 2} + l\ev_2 = \{i~:~i\cdot(1,\ldots,1) =  l + 1, i\cdot\ev_1 = 0, i\cdot\ev_2 = r - 2\}.
\end{align}
The symbols from the second systematic node which remain unmatched at the end of first stage (and require downloading additional symbols) due to their participation in the parity symbols indexed by the set $\Zc^{\{1, 2\}}_{l, (0 \rightarrow l +2)} \cup \widetilde{\Zc}^{\{2, 3\}}_{l, (2 \rightarrow l + 2)}$ in the $(l + 1)$-th parity node are as follows. 
\begin{align}
\label{eq:R32}
\Rc_{1 \rightarrow l} = \{i~:~i + l\ev_1 \in \Zc^{\{1, 2\}}_{l, (0 \rightarrow l +2)}\} =  \{i~:~i\cdot(1,\ldots, 1) = 2, i\cdot \ev_1 = r - l, i \cdot \ev_2 \neq 0\}. \\
\widetilde{\Rc}_{1 \rightarrow l} = \{i~:~i + l\ev_1 \in\widetilde{\Zc}^{\{2, 3\}}_{l, (2 \rightarrow l + 2)}\} =  \{i~:~i\cdot(1,\ldots,1) =  1, i\cdot\ev_1 = r - l, i\cdot\ev_2 = r - 2\}. \label{eq:R321}
\end{align}

\item \textbf{Unmatched symbols from the third systematic node:}~The symbols from the second systematic node that are (potentially) matched using the parity symbols downloaded from the $(l + 1)$-th parity node during the first stage are indexed by the following two sets. 
\begin{align}
\Uc^{\{2, 3\}}_{1 \rightarrow l+2} = \{i~:~i\cdot(1,\ldots, 1) \neq 0, i\cdot \ev_1 = r - (l + 2), i \cdot \ev_2 = 0\}. \\
\widetilde{\Uc}^{\{1, 3\}}_{0 \rightarrow l+3} = \{i~:~i\cdot(1,\ldots, 1) = l + 3, i\cdot \ev_1 = r - 3, i \cdot \ev_2 = 0\} \subset \Uc^{\{1, 3\}}_{0 \rightarrow l+3}.\\
\end{align}
Using \eqref{eq:Zset}, we can identify the indices of the partiy symbols from the $(l + 1)$-parity node where these symbols participate as follows.
\begin{align}
\Zc^{\{2, 3\}}_{l, (1 \rightarrow l +2)} = \Uc^{\{2, 3\}}_{1 \rightarrow l+2} + l\ev_1 &= \{i~:~i\cdot(1,\ldots, 1)  \neq l, i\cdot \ev_1 = r - 2, i \cdot \ev_2 = 0\}. \\
\widetilde{\Zc}^{\{1, 3\}}_{l, (0 \rightarrow l +3)} = \widetilde{\Uc}^{\{1, 3\}}_{0 \rightarrow l+3} &= \{i~:~i\cdot(1,\ldots, 1) = l + 3, i\cdot \ev_1 = r - 3, i \cdot \ev_2 = 0\}.
\end{align}
The symbols from the third systematic node which remain unmatched at the end of first stage (and require downloading additional symbols) due to their participation in the parity symbols indexed by the set $\Zc^{\{2, 3\}}_{l, (1 \rightarrow l +2)} \cup \widetilde{\Zc}^{\{1, 3\}}_{l, (0 \rightarrow l +3)} $ in the $(l + 1)$-th parity node are as follows. 
\begin{align}
\label{eq:R33}
\Rc_{2 \rightarrow l} = \{i~:~i + l\ev_2 \in \Zc^{\{2, 3\}}_{l, (1 \rightarrow l +2)}\} =  \{i~:~i\cdot(1,\ldots, 1) \neq 0, i\cdot \ev_1 = r - 2, i \cdot \ev_2 = r - l\}. \\
\widetilde{\Rc}_{2 \rightarrow l} = \{i~:~i + l\ev_2 \in \widetilde{\Zc}^{\{1, 3\}}_{l, (0 \rightarrow l +3)} \} = \{i~:~i\cdot(1,\ldots, 1) = 3, i\cdot\ev_1 = r - 3, i\cdot \ev_2 = r - l\}  \label{eq:R332}
\end{align}
\end{itemize}

We now describe the set of additional symbols downloaded to match the unmatched symbols from the three failed systematic nodes (cf.~\eqref{eq:R31}, \eqref{eq:R32} and \eqref{eq:R33}).
\begin{itemize}
\item \textbf{Additional symbols downloaded to match remaining symbols from the first systematic node:} Consider a set of $r - 1$ integers $\Ic_0 = \{i_{0, 1},\ldots, i_{0, r-1}\} \subset [0, r^{k-1}-1]$ such that the following two conditions hold.
\begin{enumerate}
\item $i_{0, j} \cdot(1,1,\ldots, 1) = 1~\forall~j \in [r-1]$. 
\item $i_{0, j}\cdot\ev_1 = j~\text{for}~j \in [r-1]$.
\item $i_{0, j}\cdot\ev_2 = r -1~\forall~j \in [r-1]$.
\end{enumerate}
Note that we are using vector representation of the integers from the $\Ic_0$ in $\ZZ_r^{k-1}$ in order to define these three requirements. For $j \in  [3, k-1]$, the set of additional symbols downloaded from the $(j+1)$-th systematic node in order to match the remaining symbols from the first systematic node have their row indices belonging to the following set. 
\begin{align}
\label{eq:Sdef}
\Sc^{\{1, 2, 3\}}_0 = \Ic_0 + \{a_1(\ev_4 - \ev_3) + \cdots + a_{k-4}(\ev_{k-1} - \ev_3)~:~(a_1,\ldots, a_{k-4}) \in [r-1]^{k-4} \}.
\end{align}
Note that we have $|\Sc^{\{1, 2, 3\}}_{0}|= (r - 1)r^{k-4}$. Next, we identify the set of the parity symbols in the $(l + 1)$-th parity nodes where these symbols appear. Let $\Pc^{\{1, 2, 3\}}_{0, l}$ denote the indices of these parity symbols in the $(l + 1)$-th parity node. Then, from the definition of the zigzeg sets (cf.~\ref{eq:Zset}), we have that
\begin{align}
\Pc^{\{1, 2, 3\}}_{0, l} &= \Sc^{\{1, 2, 3\}}_{0} + l\ev_3 = \Sc^{\{1, 2, 3\}}_{0} + l\ev_4 = \cdots = \Sc^{\{1, 2, 3\}}_{0} + l\ev_{k-1} \nonumber \\
&= \{i~:~i\cdot(1,\ldots,1) = l+1, i\cdot\ev_1 \neq 0, i\cdot\ev_2 = r - 1\}.
\end{align}
We can again use \eqref{eq:Zset} to indentify the symbols from the second systematic node that appear in the parity symbols from $l + 1$-th parity symbols that are indexed by the set $\Pc^{\{1, 2, 3\}}_{0, l}$.
\begin{align}
\{i~:~i \in \Pc^{\{1, 2, 3\}}_{0, l}\} = \{i~:~i\cdot(1,\ldots,1) = l+1, i\cdot\ev_1 \neq 0, i\cdot\ev_2 = r - 1\}.
\end{align}
Note that this is exactly equal to $\Rc_{0 \rightarrow l + 1}$ which is the indices of the unmatched symbols from the first systematic as they appeared in the parity symbols downloaded from the $(l + 1)$-th parity node during the first stage.

\item \textbf{Additional symbols downloaded to match remaining symbols from the second systematic node:}~ \begin{enumerate}
\item Consider a set of $r - 1$ integers $\Ic_1 = \{i_{1, 1},\ldots, i_{1, r-1}\} \subset [0, r^{k-1}-1]$ such that the following two conditions hold.
\begin{enumerate}
\item $i_{1, j} \cdot(1,1,\ldots, 1) = 2~\forall~j \in [r-1]$. 
\item $i_{1, j}\cdot\ev_1 = r-1~\forall~j \in [r-1]$.
\item $i_{1, j}\cdot\ev_2 = j~\text{for}~j \in [r-1]$.
\end{enumerate}
Note that we are using vector representation of the integers from the $\Ic_1$ in $\ZZ_r^{k-1}$ in order to define these three requirements. For $j \in  [3, k-1]$, the set of additional symbols downloaded from the $(j+1)$-th systematic node in order to match the remaining symbols from the second systematic node have their row indices belonging to the following set. 
\begin{align}
\label{eq:Sdef}
\Sc^{\{1, 2, 3\}}_1 = \Ic_1 + \{a_1(\ev_4 - \ev_3) + \cdots + a_{k-4}(\ev_{k-1} - \ev_3)~:~(a_1,\ldots, a_{k-4}) \in [r-1]^{k-4} \}.
\end{align}
Note that we have $|\Sc^{\{1, 2, 3\}}_{1}|= (r - 1)r^{k-4}$. Next, we identify the set of the parity symbols in the $(l + 1)$-th parity nodes where these symbols appear. Let $\Pc^{\{1, 2, 3\}}_{1, l}$ denote the indices of these parity symbols in the $(l + 1)$-th parity node. Then, from the definition of the zigzeg sets (cf.~\ref{eq:Zset}), we have that
\begin{align}
\Pc^{\{1, 2, 3\}}_{1, l} &= \Sc^{\{1, 2, 3\}}_{1} + l\ev_3 = \Sc^{\{1, 2, 3\}}_{1} + l\ev_4 = \cdots = \Sc^{\{1, 2, 3\}}_{1} + l\ev_{k-1} \nonumber \\
&= \{i~:~i\cdot(1,\ldots,1) = l+2, i\cdot\ev_1 = r-1, i\cdot\ev_2 \neq 0\}.
\end{align}
We can again use \eqref{eq:Zset} to indentify the symbols from the second systematic node that appear in the parity symbols from $l + 1$-th parity symbols that are indexed by the set $\Pc^{\{1, 2, 3\}}_{1, l}$.
\begin{align}
\{i~:~i + l\ev_1 \in \Pc^{\{1, 2, 3\}}_{1, l}\} = \{i~:~i\cdot(1,\ldots,1) = 2, i\cdot\ev_1 = r - (l +1), i\cdot\ev_2 \neq 0\}.
\end{align}
Note that this is exactly equal to $\Rc_{1 \rightarrow l + 1}$ which is the indices of the unmatched symbols from the second systematic as they appeared in the parity symbols downloaded from the $(l + 1)$-th parity node during the first stage.

\item Consider an integers $\tilde{i}_1 \in [0, r^{k-1}-1]$ such that the following two conditions hold.
\begin{enumerate}
\item $\tilde{i}_{1} \cdot(1,1,\ldots, 1) = 1$. 
\item $\tilde{i}_{1}\cdot\ev_1 = r-1$.
\item $\tilde{i}_{1}\cdot\ev_2 = r - 2$.
\end{enumerate}
For $j \in  [3, k-1]$, we download additional symbols from the $(j+1)$-th systematic node with their row indices belonging to the following set. 
\begin{align}
\label{eq:Sdef}
\widetilde{\Sc}^{\{1, 2, 3\}}_1 = \tilde{i}_1 + \{a_1(\ev_4 - \ev_3) + \cdots + a_{k-4}(\ev_{k-1} - \ev_3)~:~(a_1,\ldots, a_{k-4}) \in [r-1]^{k-4} \}.
\end{align}
Note that we have $|\widetilde{\Sc}^{\{1, 2, 3\}}_{1}|= r^{k-4}$. Next, we identify the set of the parity symbols in the $(l + 1)$-th parity nodes where these symbols appear. Let $\widetilde{\Pc}^{\{1, 2, 3\}}_{1, l}$ denote the indices of these parity symbols in the $(l + 1)$-th parity node. Then, from the definition of the zigzeg sets (cf.~\ref{eq:Zset}), we have that
\begin{align}
\widetilde{\Pc}^{\{1, 2, 3\}}_{1, l} &= \widetilde{\Sc}^{\{1, 2, 3\}}_{1} + l\ev_3 = \widetilde{\Sc}^{\{1, 2, 3\}}_{1} + l\ev_4 = \cdots = \widetilde{\Sc}^{\{1, 2, 3\}}_{1} + l\ev_{k-1} \nonumber \\
&= \{i~:~i\cdot(1,\ldots,1) = l+1, i\cdot\ev_1 = r-1, i\cdot\ev_2 =  r - 2\}.
\end{align}
We can again use \eqref{eq:Zset} to indentify the symbols from the second systematic node that appear in the parity symbols from $l + 1$-th parity symbols that are indexed by the set $\widetilde{\Pc}^{\{1, 2, 3\}}_{1, l}$.
\begin{align}
\{i~:~i + l\ev_1 \in \widetilde{\Pc}^{\{1, 2, 3\}}_{1, l}\} = \{i~:~i\cdot(1,\ldots,1) = 1, i\cdot\ev_1 = r - (l +1), i\cdot \ev_2 = r - 2\}.
\end{align}
Note that this is exactly equal to the unmatched symbols from the second systematic node denoted by $\widetilde{\Rc}_{1 \rightarrow l + 1}$ (cf.~\eqref{eq:R321}).

\end{enumerate}

\item \textbf{Additional symbols downloaded to match remaining symbols from the third systematic node:}~
\begin{enumerate}
\item Consider a set of $r - 1$ integers $\Ic_2 = \{i_{2, 1},\ldots, i_{2, r-1}\} \subset [0, r^{k-1}-1]$ such that the following two conditions hold.
\begin{enumerate}
\item $i_{2, j} \cdot(1,1,\ldots, 1) = j~\text{for}~j \in [r-1]$. 
\item $i_{2, j}\cdot\ev_1 = r-2~\forall~j \in [r-1]$.
\item $i_{2, j}\cdot\ev_2 = r-2~\forall~j \in [r-1]$.
\end{enumerate}
Note that we are using vector representation of the integers from the $\Ic_2$ in $\ZZ_r^{k-1}$ in order to define these three requirements. For $j \in  [3, k-1]$, the set of additional symbols downloaded from the $(j+1)$-th systematic node in order to match the remaining symbols from the third systematic node have their row indices belonging to the following set. 
\begin{align}
\label{eq:Sdef}
\Sc^{\{1, 2, 3\}}_2 = \Ic_2 + \{a_1(\ev_4 - \ev_3) + \cdots + a_{k-4}(\ev_{k-1} - \ev_3)~:~(a_1,\ldots, a_{k-4}) \in [r-1]^{k-4} \}.
\end{align}
Note that we have $|\Sc^{\{1, 2, 3\}}_{2}|= (r - 1)r^{k-4}$. Next, we identify the set of the parity symbols in the $(l + 1)$-th parity nodes where these symbols appear. Let $\Pc^{\{1, 2, 3\}}_{2, l}$ denote the indices of these parity symbols in the $(l + 1)$-th parity node. Then, from the definition of the zigzeg sets (cf.~\ref{eq:Zset}), we have that
\begin{align}
\Pc^{\{1, 2, 3\}}_{2, l} &= \Sc^{\{1, 2, 3\}}_{2} + l\ev_3 = \Sc^{\{1, 2, 3\}}_{2} + l\ev_4 = \cdots = \Sc^{\{1, 2, 3\}}_{2} + l\ev_{k-1} \nonumber \\
&= \{i~:~i\cdot(1,\ldots,1) \neq l, i\cdot\ev_1 = r-2, i\cdot\ev_2 = r - 2\}.
\end{align}
We can again use \eqref{eq:Zset} to indentify the symbols from the third systematic node that appear in the parity symbols from $(l + 1)$-th parity symbols that are indexed by the set $\Pc^{\{1, 2, 3\}}_{2, l}$.
\begin{align}
\{i~:~i + l\ev_2 \in \Pc^{\{1, 2, 3\}}_{2, l}\} = \{i~:~i\cdot(1,\ldots,1) \neq 0, i\cdot\ev_1 = r - 2, i\cdot\ev_2 = r - (l+2)\}.
\end{align}
Note that this is exactly equal to $\Rc_{2 \rightarrow l + 2}$ which is the indices of the unmatched symbols from the third systematic node as they appeared in the parity symbols downloaded from the $(l + 1)$-th parity node during the first stage.

\item Consider an integers $\tilde{i}_2 \in [0, r^{k-1}-1]$ such that the following two conditions hold.
\begin{enumerate}
\item $\tilde{i}_{2} \cdot(1,1,\ldots, 1) = 3$. 
\item $\tilde{i}_{2}\cdot\ev_1 = r-3$.
\item $\tilde{i}_{2}\cdot\ev_2 = r - 1$.
\end{enumerate}
For $j \in  [3, k-1]$, we download additional symbols from the $(j+1)$-th systematic node with their row indices belonging to the following set. 
\begin{align}
\label{eq:Sdef}
\widetilde{\Sc}^{\{1, 2, 3\}}_2 = \tilde{i}_2 + \{a_1(\ev_4 - \ev_3) + \cdots + a_{k-4}(\ev_{k-1} - \ev_3)~({\rm mod}~r)~:~(a_1,\ldots, a_{k-4}) \in [r-1]^{k-4} \}.
\end{align}
Note that we have $|\widetilde{\Sc}^{\{1, 2, 3\}}_{2}|= r^{k-4}$. Next, we identify the set of the parity symbols in the $(l + 1)$-th parity nodes where these symbols appear. Let $\widetilde{\Pc}^{\{1, 2, 3\}}_{2, l}$ denote the indices of these parity symbols in the $(l + 1)$-th parity node. Then, from the definition of the zigzag sets (cf.~\ref{eq:Zset}), we have that
\begin{align}
\widetilde{\Pc}^{\{1, 2, 3\}}_{2, l} &= \widetilde{\Sc}^{\{1, 2, 3\}}_{2} + l\ev_3 = \widetilde{\Sc}^{\{1, 2, 3\}}_{2} + l\ev_4 = \cdots = \widetilde{\Sc}^{\{1, 2, 3\}}_{2} + l\ev_{k-1} \nonumber \\
&= \{i~:~i\cdot(1,\ldots,1) = l+3, i\cdot\ev_1 = r-3, i\cdot\ev_2 =  r - 1\}.
\end{align}
We can again use \eqref{eq:Zset} to indentify the symbols from the third systematic node that appear in the parity symbols from $l + 1$-th parity symbols that are indexed by the set $\widetilde{\Pc}^{\{1, 2, 3\}}_{2, l}$.
\begin{align}
\{i~:~i + l\ev_2 \in \widetilde{\Pc}^{\{1, 2, 3\}}_{2, l}\} = \{i~:~i\cdot(1,\ldots,1) = 3, i\cdot\ev_1 = r - 3, i\cdot \ev_2 = r - (l + 1)\}.
\end{align}
Note that this is exactly equal to the unmatched symbols from the third systematic node denoted by $\widetilde{\Rc}_{2 \rightarrow l + 1}$ (cf.~\eqref{eq:R332}).

\end{enumerate}

\end{itemize}

\end{document}